\documentclass{llncs}
\usepackage[margin = 1.18in]{geometry}
\usepackage[utf8]{inputenc}
\usepackage{amsmath}
\usepackage{amssymb, amsfonts, cite}
\usepackage{authblk}
\usepackage{ifthen}
\usepackage[ruled]{algorithm2e}
\usepackage{comment}
\usepackage{url}
\usepackage{float}

\usepackage{pgf}
\usepackage{tikz}
\usetikzlibrary{arrows,automata,positioning}
\newcommand{\altura}{.45cm}

\usepackage{microtype}

\usepackage{hyperref}
\usepackage{multirow, multicol}
\usepackage{subfigure, graphicx} %epsfig,
\usetikzlibrary{arrows,shapes}

\tikzset{circle/.pic={
\node[circle, aspect=1, draw, minimum size=0.3cm, text width=0.2cm] () at (0,0) {\tikzpictext};
}} %inner sep=0pt, minimum size=0.6cm,
\usepackage[sectionbib]{natbib}
\usepackage{microtype}
\usepackage{pgfplots}
\usepackage{booktabs}
\pgfplotsset{compat=1.10}

\bibpunct{[}{]}{;}{a}{,}{;}
\begin{document}

\title{Wealth Inequality and the Price of Anarchy}
% \author[1]{Kurtulu\c{s} Gemici}
% \author[2]{Elias Koutsoupias}
% \author[3]{Barnab\'e Monnot}
% \author[4]{Christos H. Papadimitriou
% \author[3]{Georgios Piliouras}
% \affil[1]{Department of Sociology, National University of Singapore \email{kgemici@nus.edu.sg}}
% \affil[2]{Department of Computer Science, University of Oxford \email{elias@cs.ox.ac.uk}}
% \affil[3]{Engineering Systems \& Design, Singapore University of Technology and Design \email{monnot\_barnabe@mymail.sutd.edu.sg, georgios@sutd.edu.sg}}
% \affil[4]{Department of Computer Science, Columbia University, USA \email{christos@cs.columbia.edu}}
\author{Kurtulu\c{s} Gemici\inst{1} \and Elias Koutsoupias\inst{2} \and Barnab\'e Monnot\inst{3} \and \\\vspace{-10px} Christos Papadimitriou\inst{4} \and Georgios Piliouras\inst{5}}
\institute{Department of Sociology, National University of Singapore \email{kgemici@nus.edu.sg} \and Department of Computer Science, University of Oxford \email{elias@cs.ox.ac.uk} \and Engineering Systems \& Design, Singapore University of Technology and Design \email{monnot\_barnabe@mymail.sutd.edu.sg} \and Department of Computer Science, Columbia University, USA \email{christos@cs.columbia.edu} \and Engineering Systems \& Design, Singapore University of Technology and Design \email{georgios@sutd.edu.sg}}
%\terms{congestion games, inequality}

\maketitle

\begin{abstract}
Price of anarchy quantifies the degradation of social welfare in games due to the lack of a centralized authority that can enforce the optimal outcome. At its antipodes, mechanism design studies how to ameliorate these effects by incentivizing socially desirable behavior (e.g., via tolls or taxes) and implementing the optimal state as equilibrium. In practice, the responsiveness to such measures is not uniform across the population but instead depends on the wealth of each individual. This leads to a natural, but largely unexplored, question. Does optimal mechanism design entrench, or maybe even exacerbate, social inequality?

We study this question in nonatomic congestion games, arguably one of the most thoroughly studied settings from the perspectives of price of anarchy as well as mechanism design. We introduce a new model that incorporates the wealth distribution of the population and captures the income elasticity of travel time (i.e., how does loss of time translate to lost income). This allows us to argue about the equality of wealth distribution both before and after employing a mechanism. We start our analysis by establishing a broad qualitative result, showing that tolls always increase inequality in symmetric congestion games under any reasonable metric of inequality, e.g., the Gini index. Next, we introduce the iniquity index, a novel measure for quantifying the magnitude of these forces towards a more unbalanced wealth distribution and show it has good normative properties (robustness to scaling of income, no-regret learning). We analyze iniquity both in theoretical settings (Pigou's network under various wealth distributions) as well as experimental ones (based on a large scale field experiment in Singapore). Finally, we provide an algorithm for computing optimal tolls for any point of the trade-off of relative importance of efficiency and equality. We conclude with a discussion of our findings in the context of theories of justice as developed in contemporary social sciences and present several directions for future research.
\end{abstract}

\section{Introduction}
\label{sec:intro}
Inequality in wealth and income have been rampant worldwide in the past four decades
\citep{piketty2014capital,stiglitz2012price}, considered by many the scourge of modern societies. Economic analysis,
on the other hand, traditionally focuses on efficiency, that is to say, Pareto optimality of the
allocation.  Whether, and to what extent, efficiency and equality are at loggerheads has been debated in economics,
and the verdict appears to depend on context and assumptions.
% see the discussion in Section \ref{sec:conclusion}.

Modern societies also give rise to a plethora of strategic scenaria,
in which the behavior of one agent affects the others, and the
outcome of which ultimately affects the agents' overall well-being.  In game theory, we
study the efficiency of these strategic situations through the so-called {\em
price of anarchy,} the relative efficiency of the game's Nash equilibria over the
social optimum \citep{KoutsoupiasP99WorstCE}. To combat the price of anarchy, the introduction of {\em tolls}
which enforce the optimal outcome as equilibrium has been proposed, see \citep{cole2003pricing,fleischer2004tolls} among an extensive literature.
However, the effect that such mechanisms may have on
the level of inequality in the society does not appear to have been addressed in the literature.

{\em The present paper is a first attempt to articulate and study this issue.}
We consider games (here only congestion games) in which the agents'
utility and behavior depend explicitly on their income or wealth, and study
the effect the game's equilibria have on inequality.

\paragraph{Example: Transportation in Singapore,} seen as a congestion game
with tolls, has a price of anarchy that is close to one \citep{monnot2017routing}.
The main arteries are almost never clogged, and public transportation is accessible and runs
smoothly. This is the result of bold policy decisions:  car ownership in Singapore is
significantly taxed, and  dynamically adaptive tolls are in place.
Interestingly, transportation delays seem to be a decreasing function of income (see Section \ref{sec:data}
on data). Could this be indicative of a more intrinsic tension between efficiency and equality?

We are interested in the ways in which optimal (or more generally efficiency-enhancing) mechanisms affect inequality.
Inequality is measured in many ways, but perhaps most often through
what is known as the {\em Gini coefficient (or Gini index)}.\footnote{Gini is the predominant way of measuring inequality, but there are others -- e.g. Hoover coefficient,
variance-to-mean, etc.}
 Intuitively, the Gini coefficient of a distribution
is {\em twice} the area between the $45^o$ line and the
normalized convex cumulative wealth/income curve (see Figure \ref{fig:ginilc}).  That is, we compute the cumulative
income/wealth $Q(y)$ of the lowest $y$ fraction of the population for all $0\leq y\leq 1$,
we normalize it so that $Q(1)=1$, and then we integrate $y-Q(y)$ from $0$ to $1$.  At total equality the Gini is zero, while
at total inequality (i.e., when the emperor owns everything) it is one.  In 2015, the income Gini in
OECD countries ranged from the .20s (Northern Europe) to the .40s and .50s (USA and East Asia).
Unsurprisingly, a lot of emphasis is put on the trends of the Gini itself.\footnote{In Singapore, the Gini coefficient has fallen from .470 in 2006 to .458 in 2016, whereas after accounting for government taxes and transfers, the 2016 figure was even lower at .402 \citep{archive2018}.}  Namely, is inequality increasing, and if so, at what rate?

\def\inequalityeffect{iniquity}
\def\Inequalityeffect{Iniquity}

{\bf Our contributions.}
We study this question in nonatomic congestion games with tolls, where  we introduce a new model that incorporates the wealth distribution of the population and captures the income elasticity of travel time (i.e., how loss of time translates to lost income). This allows us to argue about the equality of wealth distribution both before and after employing a mechanism.
The basics of our modeling are thus:  We consider a continuum of agents, each agent of a type $x>0$
standing for their income.\footnote{Or wealth; we write ``income'' henceforth in this paper, but ``wealth'' would
also be appropriate everywhere.} We assume that the distribution of types is known.
When these agents engage in a game $\Gamma$ and that, at equilibrium, type $x$ receives a cost
$c_x$. This cost is expressed in the same units as income, dollars, say; after incorporating the losses due to time spent in traffic in dollars as well as any possible costs due to tolls/taxes.
%note that $u_x$ will in general be negative),
As a result, the agent's total wealth becomes $x' = x - \alpha c_x$,
where $\alpha$ is a small constant standing for the importance of the game under consideration
to an individual's well-being.
%We only consider the case of constant $\alpha$ independent of $x$
%--- obviously this assumption is not always warranted (for example, commuting may affect
%the poor more than the rich), an important problem we leave open here.
In Section \ref{sec:inequality} we establish a broad qualitative result, the Iniquity Theorem (Theorem \ref{thm:iniquity}), showing that tolls always increase inequality in symmetric congestion games under any reasonable metric of inequality, e.g., Gini coefficient.
% with the Gini coefficient being our primary focus.

In Section \ref{sec:iniquity} we introduce the iniquity index, a novel measure for quantifying the magnitude of these forces towards a more unbalanced wealth distribution. Let $q$ be the initial income distribution of the population
of agents under consideration, let $G(q)$ be its Gini coefficient, and suppose that $\hat{q}$ is the
distribution of the income {\em after} each $x$ becomes $x-\alpha\cdot c_x$
(that is, after the game has been played).
We are interested in the way the game affects the Gini coefficient; we express this, informally,
as the coefficient of $\alpha$ in $G(\hat{q})-G(q)$, ignoring terms that are $o(\alpha)$; in other words,
we are interested in the {\em derivative} of $G(q)$ with respect to $\alpha$.  We call this quantity
the {\em {\inequalityeffect}}\ of the game. We show that from a theoretical perspective it has attractive properties.
Specifically it is robust to scaling of income (Theorem \ref{thm:scaleinv}) and it remains unaffected if instead of immediate equilibration we assume
that all agents apply regret-minimizing algorithms (Theorem \ref{thm:noregret}).

% Positive \inequalityeffect\ means that the
%current equilibria of the game create a greater Gini and thus a
%higher level of inequality.  One of the main messages of this work is {\em  an explicit trade-off
%between \inequalityeffect\ and price of anarchy.}  For example, an immediate result following from these
%definitions is the following (Theorem \ref{thm:iniquity}):  The \inequalityeffect\ of any symmetric congestion network
%with tolls in which there is an equilibrium with positive flows through tolls (equivalently, with two paths
%with different delays), is always positive.
%In fact, this is true not only for the Gini coefficient, but with respect to any ``reasonable''
%measure of inequality.  Without tolls, inequality remains unaffected.

 We analyze iniquity both in theoretical settings (Section \ref{sec:pigou})  as well as experimental ones (Section \ref{sec:data}).
Specifically, these effects become apparent already in the well-trodden
Pigou's network \citep{pigou2013economics}. This network has two parallel links, one with constant delay function 1,
 and another with delay function $x$
(that is, a delay proportional to the percentage of agents that take this option). This game {\em without tolls}
has a unique Nash equilibrium in which all flow goes through the $x$ link, while the optimum would be
to split equally between the two. The price of anarchy is $4\over 3$, and the \inequalityeffect\  turns out to be zero.
It is well-known that the price of anarchy, {\em in the case of equal incomes,}
can be rendered to one by adding tolls, and
it is not hard to see that the same can be done for any income distribution \citep{cole2003pricing}
--- {\em but then the \inequalityeffect\  becomes substantial.}
If tolls decrease, we have a full-fledged trade-off between \inequalityeffect\  and price of anarchy.
In Theorem \ref{thm:pigou} we calculate the precise price of anarchy to \inequalityeffect\
trade-off of any variant of Pigou's network with income distributions of the form $y^\beta$.

In Section \ref{sec:data} we perform detailed data analytics on a semantically rich dataset capturing the routing
behavior of tens of thousands of Singaporean students. This dataset captures the movement of each individual at a high frequency (one new datapoint per individual every 13 seconds) and allows us to distinguish between different modes of transportation (walking, bus, train, car). We can pinpoint each individual's home location which allows us to
 compute estimates about the wealth distribution of the participants. Given the level of data granularity, we can control for different parameters (e.g., distance between source and sink destinations for different sub-populations) and identify a statistically significant increased commute time for the lower-income students, which corroborates our theoretical analysis. Interestingly, the Singapore case also points out some of the successful policies (e.g., polycentric urban development model) that can be implemented to alleviate the trade-off between efficiency and equality. The main part of the paper concludes with a discussion of our findings in the context of theories of justice in contemporary social sciences (Section \ref{distributive-justice}) and presents several directions for future research.

In Appendix \ref{sec:asymmetric} we present some surprising results about iniquity in (adversarially chosen) asymmetric settings, which prompt several open questions. In Appendix \ref{sec:tradeoff}  we provide an algorithm for computing optimum tolls for any point of the trade-off of relative importance of efficiency and equality for symmetric networks on parallel links and arbitrary delay functions. %The algorithm is rather straightforward, giving us hope that the results can be extended to symmetric networks and beyond.

% We revisit the classical results on the price of anarchy in non-atomic
%congestion games (with tolls) through this lens.  We conclude that nominal efficiency very much begets inequality.

\begin{comment}
How about more general networks and distributions?
We consider in Section \ref{sec:tradeoff} symmetric networks with a fixed number of parallel links
and arbitrary delay functions, and a continuum of (non-atomic) agents with a given income distribution.
(we consider various ways of presenting the income distribution, including a list of ``generalized percentiles''
($\epsilon$-strong segments of contiguous income values). We show (Theorem \ref{thm:tradeoff}) that, for
any $\lambda>0$ capturing the relative importance of equality, we can calculate the congestions and tolls
that will optimize the sum of total delays plus $\lambda$ times the resulting Gini coefficient (among many
other similar measures and parameters).  The algorithm is rather straightforward, giving us hope that the
results can be extended to symmetric networks and beyond.
%
We conclude with discussion and a panorama of open questions.
\end{comment}

\section{Related Work}

Given the proliferation of the usage of algorithms in all aspects of our lives (from suggesting Airbnb hosts to identifying convicts eligible for early parole), the theoretical computer science community has recently
focused with zeal on understanding
% pinpointed a number of interesting and challenging open questions involving
issues of  fairness, equality and justice. Interestingly, it seems that despite this fervor of activity,
the questions explored in this paper have so far received little attention.
Despite this distance in subject matter, it is useful to examine this rapidly forming research landscape.
 We provide some key references below:

{\bf Algorithmic Fairness, Fairness in AI.} This line of research tries to adapt standard learning and decision-making settings so as the new algorithms satisfy some fairness desiderata.
\citet{Kannan:2017:FIM:3033274.3085154} investigate whether it is possible to design  payment schemes for a principal to motivate myopic agents to play fairly in  almost all rounds. Their notion of fairness asks that more qualified individuals are never  preferred over less qualified ones.
\citet{JabbariJKMR17} study fairness for reinforcing learning algorithms over Markovian environments, where their fairness constraint requires that an algorithm never prefers one action over another if the long-term (discounted) reward of choosing the latter action is higher.  \citet{pleiss2017fairness} study the tension between minimizing error disparity across different population groups while maintaining calibrated probability estimates.

{\bf Cake Cutting, Fair Division.} This line of research asks how to divide some goods (or obligations e.g. rent) according to different fairness principles. This is a classic problem \citep{steinhaus1948problem}, a good reference point for which is
the book by \citet{brandt2016handbook}. Recently, many such fair division algorithms have been made freely available online on a nonprofit website \citep{goldman2015spliddit}. The availability of such data has fed back into theoretical research leading to new methods for the classic rent division problem \citep{gal2017fairest}.

{\bf Price of Fairness.} \citet{bertsimas2011price,Bertsimas:2012:ET:2421296.2421303} have coined the term price of fairness to quantify the efficiency loss in optimization settings where the set of permissible outcomes must meet some equality or fairness constraints. This line of research stays close to the philosophy and techniques of the price of anarchy literature. In contrast,  our work although inspired by price of anarchy moves in a different direction where new techniques and ideas are required.

{\bf Smart Cities: Singapore.} Singapore takes a proactive approach towards using data analytics to understand and improve both the efficiency as well as the fairness of its practices. \citet{benabbou2017diversity} has used data to develop and test new techniques that promote diversity in the allocation of public housing (e.g. every ethnic group must not own more than a certain percentage in a housing project). \citet{monnot2017routing} used data analytics to compute bounds on the price of anarchy as well as other metrics of efficiency of the traffic in Singapore (regret, equilibration).

Finally, \citet{chen2017balancing,chen2017mechanism} have recently used the Gini-coefficient over the probabilities of the agent winning probabilities as an inequality measure of different mechanisms and design mechanisms with such good properties. Although syntactically similar, these works do not model wealth distributions nor do they examine the differential effects of mechanisms to equality, which is our focus. Our work shows, perhaps unsurprisingly, that the use of public transportation plays a critical, but not well-understood, role in the functioning of a traffic network.  \citet{fotakis2017selfish} recently introduced a model of congestion games with buses, and hopefully more research will follow along these lines.

\section{Model Description}
\label{sec:prelim}
\newcommand{\ii}{q}
\newcommand{\epii}{\hat{q}}
\newcommand{\II}{Q}
\newcommand{\cost}{\text{cost}}
\newcommand{\Gini}{G}
\newcommand{\latency}[1]{\ell_{#1}}
\newcommand{\flow}[1]{c(#1)}
\newcommand{\s}{c_u} % switching point
\newcommand{\toll}{\tau}
\newcommand{\iniquity}{I}
\def\cost{{\rm cost}}

We  describe a game-theoretic model where a continuum of agents participates in a traffic congestion game with tolls.  The total disutility for each agent depends both on their traffic-induced latency as well as on the tolls, whose effects are experienced differentially based on each agent's income level.

\paragraph{Congestion game}
A {\em symmetric congestion game with type-specific costs} consists of a finite set $E$ of edges,
and a finite subset of $2^E$ called the set of \textit{paths} $\mathcal{P}$, common to all types.
We shall only deal with {\em network} congestion games, where the set of paths
consists of all possible paths between two nodes $s$ and $t$ in a graph with edge set $E$.

\paragraph{Income}
We have a continuum of {\em types} which lie in $[0,1]$.
Type $x$ has income $\ii(x)$, where $\ii$ is the {\em quantile function} of the income
of a population of agents --- that is, \( |z: \ii(z) \leq \ii(x)| = x \), where \( |\cdot| \) is the Lebesgue measure.
We shall further assume that $\ii(0)>0$ and $\ii$ is measurable and nondecreasing. Typically, we will assume a continuum of types and a strictly increasing, continuous $q$.
 In this case, if we treat income as random variable, then $q$ expresses the inverse of its cumulative distribution function.

\paragraph{Edge cost functions}
Our main result, the Iniquity Theorem, holds under general conditions on the cost functions. Assumptions on their shape and properties are stated here while further specific examples will be given later in this Section. Each agent \( x \) using edge \( e \) experiences a cost \( f_e(\ii, z, \toll_e) \), where \( \ii \) is the agent's income, \( z \) is the level of congestion on edge \( e \) and \( \toll_e \) is the fixed toll paid by the agent.

The following assumptions on the edge cost functions~\eqref{eq:edgecf} are made throughout:
\begin{align}
    \text{The functions } (f_e)_{e \in E} \text{ are } &- \text{nonnegative} \nonumber \\
      &- \text{nonincreasing in income} \tag{ECF}\label{eq:edgecf}\\
      &- \text{nondecreasing in congestion and toll} \nonumber
\end{align}
In other words, for a fixed toll and congestion level, agents with higher income perceive a lower cost than agents with lower income. %, similar to routing games with type-specific costs.

\paragraph{Flow}
A {\em flow} $F: [0,1]\mapsto \mathcal{P}$ is a mapping from types to paths; we shall only need to consider
{\em finitary} flows, that is, flows $F$ which divide $[0,1]$ into finitely many intervals, and map the interiors of those
intervals to one path in $\mathcal{P}$; that is, $F$ is specified by a finite number of reals $a_0=0<a_1<a_2<\cdots<a_k=1$
such that $F(b)=F(c)$ for all $i$ and $b,c\in (a_i,a_{i+1}]$

\paragraph{Agent cost}
Let $F$ be a flow.  The {\em congestion} of this flow, $c^F$, is a function mapping $E$ to the nonnegative reals,
where $c^F(e)= |\{x: e\in F(x)\}|$, where $|\cdot|$ denotes the Lebesgue measure.
The {\em agent cost} under flow $F$ to an agent of type $x$ is some function \( C \) of her income and the path cost:
\[
  \cost^F(x) = C \Big( \ii(x), \sum_{e\in F(x)} f_e(\ii(x), c^F(e), \toll_e) \Big) \, .
\]

The model allows for some degree of flexibility when designing the overall cost of the agents. We prove our result in Section \ref{sec:inequality} for two specific forms of \( \cost^F(x) \):
\begin{align}
  & \cost^F(x) = \sum_{e\in F(x)} f_e(\ii(x), c^F(e), \toll_e) \tag{CF1} \label{eq:cf1}\\
  & \cost^F(x) = \ii(x) \cdot \sum_{e\in F(x)} f_e(\ii(x), c^F(e), \toll_e) \tag{CF2} \label{eq:cf2}
\end{align}
Note that the second form is identical to the first with the exception that costs are scaled by the income of the agent.

Within the class \( f_e \) of functions satisfying~\eqref{eq:edgecf}, a natural special case is that of \( f_e(\ii, z, \toll_e) = \frac{\toll_e}{\ii} + \latency{e}(z) \)
where \( \latency{e} \) is a nonnegative and nondecreasing latency function. Under this form, the cost of agent \( x \) in edge \( e \) with~\eqref{eq:cf2} is
\begin{equation}
  \ii(x) \cdot f_e(\ii(x), c^F(e), \toll_e) =
    \ii(x) \cdot \latency{e}(c^F(e)) + \toll_e \tag{CAN} \label{eq:canonical} \, .
\end{equation}
We call this cost \textit{canonical} as it is given in units of the toll, i.e. money.

There is an extensive discussion in the transportation literature of the true cost of transportation to the traveler and the value of time, see \citep{abrantes2011meta, borjesson2012income} for some of the most recent papers, with dozens of references therein. This field has established and studied the {\em income elasticity} of the value of (travel) time (informally, the precise nature of the formula $\toll_e \over q$ above) and validated and measured it through extensive surveys and other studies over three decades.  The upshot is that the cross-sectional elasticity (that is, the elasticity with regressive corrections across causal parameters such as having children and living in the capital) is constant across long periods of time, and that the precise relationship seems to be  $\toll_e \over \ii^\beta$ where $\beta\leq 1$ is conventionally taken to be one, even though certain countries, such as the UK, use value 0.8.

\paragraph{Nash equilibrium}
We say that a flow $F$ is a {\em Nash equilibrium} in our model if
\begin{equation}
  \tag{NE}
  \label{eq:nasheq}
  \text{For all types } x,   \text{and for all path } P \in \mathcal{P}, \, \cost^F(x) \leq C(\ii(x), \sum_{e\in P} f_e(\ii(x), c^F(e), \toll_e))
\end{equation}
that is, if no type $x$ would be better off by deviating to another path $P \in \mathcal{P}$.

\paragraph{Gini coefficient}
The Gini coefficient is the most commonly used measure of inequality.  A Gini coefficient equal to zero corresponds to perfect equality (everyone has the same income), whereas a Gini coefficient of one corresponds to maximal inequality (e.g., one person has all the income). The Gini coefficient has several desirable properties such as:
\begin{itemize}
\item \textbf{Scale independence.} The Gini coefficient does not change after rescaling incomes (e.g. change of units/currency).
\item \textbf{Population independence.} It does not depend on the size of the population.
\item \textbf{Anonymity.} It does not depend on the identity of the rich/poor individuals.
\item \textbf{Transfer principle.} If income (less than the difference\footnote{If the income transfer is less than the difference of their incomes, the ordering of the wealth of the users does not change.}), is transferred from a rich person to a poor person the resulting distribution is more equal (i.e. the Gini decreases).
\end{itemize}

There exists a strong relationship between the Gini coefficient and the \textit{Lorenz curve}. The {\em cumulative income} function is $ \II(t)=\int_0^t \ii(x) dx $. The Lorenz curve \( L(t) \) is then the fraction of total income held by individuals under and at quantile \( x \):
\begin{equation}
  \tag{LC}
  \label{eq:lorenz}
  L(t) = \frac{1}{\mu} \int_0^t \ii(x) dy = \frac{1}{\mu} \II(t)
\end{equation}
In that case, the Gini coefficient \( G \) is equal to \( 1 - 2 \int_0^1 L(t) dt \). We show the geometric intuition of this fact in Figure \ref{fig:ginilc}.
% The {\em Gini coefficient}  (also known as Gini ratio or Gini index) \citep{gini1921measurement} of the income distribution $ \ii(x) $ is
% \begin{equation}
%   \tag{\( \text{G}_{\text{CDF}} \)}
%     \label{eq:ginicdf}
%     2\int_0^1 \Big( x - {\II(t)\over \II(1)} \Big) dt.
% \end{equation}

% Equivalently, when the income distribution is given as a continuous probability distribution function \( p(x) \) then the Gini coefficient is equal to half of the relative mean absolute difference, $\frac{1}{2\mu}\int_{-\infty}^\infty\int_{-\infty}^\infty p(x)p(y)\,|x-y|\,dx\,dy$
% where $\mu$ is the mean of the distribution.

\begin{figure}
  \centering
  \begin{minipage}[t]{0.42\linewidth}
		\centering
    \begin{tikzpicture}[thick,scale=0.8, every node/.style={transform shape}]
      \definecolor{myred}{rgb} {0.941,0.3294,0.1921}
      \definecolor{myblue}{rgb} {0.2705,0.6627,0.8705}
      \definecolor{mygreen}{rgb} {0.7490,0.8078,0.5019}
      \tikzstyle{nome}=[anchor=west, minimum height=\altura,minimum width=9cm,text width=8.8cm]
      \draw[thick,->] (0,0) -- (0,4.5);
      \draw[thick,->] (0,0) -- (5,0);
      \node at (-0.5,4.5) {\( L(x) \)};
      \node at (5,-0.5) {Quantile \( x \)};
      \draw[thick] (0,0) -- (4,4);
      \draw (4,4) -- (4,0);
      \path[fill=mygreen, opacity=0.5] (0,0) parabola (4,4) -- (4,0) -- (0,0);
      \path[fill=myred, opacity=0.5] (0,0) parabola (4,4) -- (0,0);
      \draw[line width=0.5mm, myblue] (0,0) parabola (4,4);
      \node at (1.4,1) {\( A \)};
      \node at (3,1) {\( B \)};
    \end{tikzpicture}
    \caption{The Lorenz curve is plotted in blue. The green area is \( B = \int_0^1 L(t) dt \). The Gini coefficient is then \( G = 1 - 2B = 2A \).}
    \label{fig:ginilc}
  \end{minipage}
  \begin{minipage}[t]{0.1\linewidth}
    \begin{tikzpicture}
    \end{tikzpicture}
  \end{minipage}
	\begin{minipage}[t]{0.42\linewidth}
		\centering
    \begin{tikzpicture}
      \node[state] (1) at (0, 2) {S};
      \node[state] (2) at (4, 2) {T};
      \node (3) at (0, 0) {};
      \draw[->] (1) edge [bend left] node [above] {$\latency{u}(x)=1$} (2);
      \draw[->] (1) edge [bend right] node [below] {$\latency{d}(x)=x$} (2);
    \end{tikzpicture}
    \caption{The Pigou network}
    \label{fig:pigou}
  \end{minipage}
\end{figure}

\paragraph{Our motivating problem}
We consider how Nash equilibrium flow $F$ affects the incomes of the population. In particular, we assume that
the income of type $x$ changes from $\ii(x)$ to
$ \ii(x)-\alpha \cdot \cost^F(x)$ for some (intuitively small) $\alpha>0$.
We call the resulting income distribution $\epii(x)$. Notice that, in general, $\epii(x)$ may be different from
$\ii(x)-\alpha \cdot \cost^F(x)$, since the cost of $F$ may rearrange the order of types (recall that distributions such as $\ii(x)$ are assumed to be nondecreasing). As we shall see in the Iniquity Theorem proof of Section \ref{sec:inequality}, this turns out to never be the case and moreover the inequality increases as a result.

%We now state and prove an important result:

%\begin{theorem} {\bf The \Inequalityeffect\ Theorem:} In any Nash equilibrium of any symmetric congestion game with type-specific costs, either all paths with positive flows have zero tolls, or the \inequalityeffect\  of the game is positive.
%\end{theorem}

\section{Tolls Always Increase Inequality}
\label{sec:inequality}
Tolls can be used in congestion games so as to induce socially optimal flows (from the perspective of total cost) as Nash equilibrium \citep{cole2003pricing,fleischer2004tolls}. We next prove a general theorem showing that tolls always exacerbate societal inequality. So, in a sense to achieve optimality from the perspective of social welfare we have to pay a hidden cost in terms of fairness.

\begin{theorem}
    \label{thm:iniquity}
    {\bf The \Inequalityeffect\ Theorem:} In any Nash equilibrium of any symmetric congestion game with type-specific costs (for both \ref{eq:cf1} and \ref{eq:cf2} models of cost functions), any set of positive edge tolls $\toll_e$ increases the inequality of the population. More specifically, the Gini coefficient of the \textit{ex ante} income distribution \( \ii \) is lower than the Gini coefficient of the \textit{ex post} income distribution \( (\ii(x) - \alpha \cdot \cost^F(x))_{x \in [0, 1]} \).
    %there exists an $\alpha_0>0$ such that for all $\alpha \in (0,\alpha_0)$ the new gini index of the population after the income update $\ii(x)$ to $\ii(x)-\alpha \cost^F(x)$ is larger than its original value. In other words, adding taxation increases inequality in the population.
\end{theorem}

The proof is done in three steps. First, we show that if two income distributions with equal means cross at one point, one has a higher Gini coefficient than the other. This is equivalent to the transfer principle, or Pigou-Dalton principle of income inequality measures. Second, we show that when a distribution is obtained by decreasing proportionally less the higher incomes than the lower incomes --- in other words, a regressive tax --- then the resulting distribution has a higher Gini coefficient than the original one, i.e., is more unequal. Third, we show that under equilibrium in the game, players with higher incomes have a lower cost than players with lower incomes. Finally, Theorem \ref{thm:iniquity} is obtained as a corollary of the three lemmas.

\newcommand{\iid}{q}
\newcommand{\iidb}{\epii}
\begin{lemma}
	\label{lem:lorenzdec}
	Suppose \( \iid %\simeq D
	 \) and \( \iidb 
	 %\simeq \hat{D}
	  \) are two income distributions (represented by their quantile functions) 
	of equal means, i.e.
	\( \mu = \int_0^1 \iid(x)dx = \int_0^1 \iidb(x)dx = \hat{\mu} \). If there exists \(x^*\) such that \(  \iidb(x) \leq \iid(x), \forall x \leq x^* \), and \( \iidb(x) > \iid(x) \) otherwise,
	then % \( L(x) \geq \hat{L}(x),\, \forall x \) --- where \( L(x) \) is the Lorenz curve of \( D \) defined in~\eqref{eq:lorenz} --- and thus
	 \( G( \iid) \leq G( \iidb) \).
\end{lemma}
\begin{proof}
        We will show that \( L(x) \geq \hat{L}(x),\, \forall x \) where \( L(x) \) is the Lorenz curve of \(  \iid \) defined in~\eqref{eq:lorenz} and thus \( G( \iid) \leq G( \iidb) \).
	First note that for \( x \leq x^* \),
	\[
		\hat{L}(x) \leq L(x)
	\]
	since \( \iid(x) \geq \iidb(x) \).

	For \( x \geq x^* \), we have
	\begin{align*}
		\hat{L}(x) \leq L(x) & \iff & \int_0^{x^*} \iidb(y) dy + \int_{x^*}^x \iidb(y) dy \leq \int_0^{x^*} \iid(y) dy + \int_{x^*}^x \iid(y) dy \\
		& \iff & \int_0^{x^*} [\iid(y)-\iidb(y)] dy \geq \int_{x^*}^x [\iidb(y)-\iid(y)] dy
	\end{align*}
	
	The last inequality is true since by \( \mu = \hat{\mu} \) we get
	\[
		\int_0^{x^*} [\iid(y)-\iidb(y)] dy = \int_{x^*}^{1} [\iidb(y) - \iid(y)] dy
	\]
\end{proof}

\begin{lemma}
	\label{lem:giniscale}
	Suppose two income distributions  (represented by their quantile functions)  \( \iid
	% \simeq D_0 
	\) and \( \iidb 
	%\simeq D_1
	 \)  are such that \( \iidb(x) = \beta(x) \cdot \iid(x) \) and \( 1 \geq \beta(y) \geq \beta(z)>0 \) for \( y \geq z \),\footnote{I.e., \( \iidb \) is obtained from
  \( \iid \) by a transformation that reduces lower incomes relatively more than higher incomes. Income \textit{order is preserved} and \( \hat{\mu} \leq \mu \).} then %the Gini coefficients are such that
   \( G(\iid) \leq G(\iidb) \).
\end{lemma}
\begin{proof}
         If we rescale distribution $\iidb$ to 
	%Transform again \iidb  incomes to distribution % \( \iid_2 \simeq D_2 \) where 
	%with
	 \( \iid_1(x) = \frac{\mu}{\hat{\mu}} \iidb(x) \) the Gini remains invariant. It suffices to show that  \( G(\iid) \leq G(\iid_1) \).
	 Since \( \mu = \mu_1 \),  we can compare the two distributions using Lemma \ref{lem:lorenzdec}.

	Introduce \( \beta_1(x) = \frac{\mu}{\hat{\mu}} \beta(x) \), the transformation of income from \( \iid \) to \( \iid_1 \). 
	\( \beta_1(x) \) is a nondecreasing function of \( x \). If \( \beta_1(x) < 1 \) for all \( x \), then obviously we cannot have \( \mu_0 = \mu_1 \) and the same holds if \( \beta_1(x) > 1 \) for all \( x \). Hence there exists some \( x^* \) such that \( \beta_1(x^*) = 1 \). Moreover,  by the monotonicity properties of $\beta_1$, the income order is  preserved and \( x^* \) satisfies all the properties of Lemma \ref{lem:lorenzdec}. Thus,  \( G(\iid) \leq G(\iid_1) = G(\iidb)\). 
\end{proof}

\begin{lemma}
  \label{lem:decreasingcosts}
  Let \( 0 \leq x \leq y \leq 1 \) and \( F \) be an equilibrium flow.
  Let \( \cost^F(z) = \sum_{e\in F(z)} f_e(\ii(z), c^F(e), \toll_e) \) where \( f_e \) is the cost function on edge \( e \), nonincreasing in its first argument (income) and nondecreasing in the other two (congestion and toll) as per assumptions~\eqref{eq:edgecf},
  then \( \cost^F(x) \geq \cost^F(y) \).
\end{lemma}
\begin{proof}
  Suppose \( \cost^F(x) < \cost^F(y) \). Since the edge cost functions are nonincreasing functions of the income, the richer agent of type $y$ can deviate to the path chosen by agent of type $x$ and strictly decrease her cost, contradiction.
%   there exists \( a \) and \( b \) such that \( y \in (a, b] \) and \( F(z) = F(z') \) for \( z, z' \in (a, b] \). It is profitable at least for users in the interval \( [y, b] \) to unilaterally deviate to the path chosen by user \( x \) since the cost function on edges is decreasing in the income, violating the argument that \( F \) is a Nash equilibrium flow as defined in~\eqref{eq:nasheq}.
\end{proof}

We can now prove the Iniquity Theorem.

\begin{proof}{\textbf{The Iniquity Theorem:}}
    In Lemma \ref{lem:giniscale}, if \( \iid \) is our initial income distribution, \textit{before the game is played}, then \( \iidb \) is our income distribution \textit{after the game is played}. Indeed, by Lemma \ref{lem:decreasingcosts}, smaller incomes have an absolute higher cost, and thus the transformation from \( \iid \) to \( \iidb \) is obtained for some nondecreasing \( \beta \) satisfying the above conditions.

Indeed, for the \ref{eq:cf2} model, the \textit{ex post} income distribution
\[
  \Big( \ii(x) - \alpha \cdot \ii(x) \cdot \sum_{e\in F(x)} f_e(\ii(x), c^F(e), \toll_e) \Big)_x
\]
where the cost of using a path is scaled by the income of the player. In that case, a simple factorization by \( \ii(x) \) yields the \textit{ex post} distribution
\[
  \Big( \ii(x) (1 - \alpha \cdot \sum_{e\in F(x)} f_e(\ii(x), c^F(e), \toll_e)) \Big)_x \, .
\]
Since we have shown in Lemma \ref{lem:decreasingcosts} that \( \sum_{e\in F(x)} f_e(\ii(x), c^F(e), \toll_e) \) is decreasing in \( x \),
this is the setting of Lemma \ref{lem:giniscale} where incomes are scaled by a function \( \beta(x) \) nondecreasing in \( x \). The Iniquity Theorem follows. For the \ref{eq:cf1} model the analogous calculations yield \textit{ex post} distribution
 \[ \Big( \ii(x) (1 - \alpha \cdot\frac{ \sum_{e\in F(x)} f_e(\ii(x), c^F(e), \toll_e))}{\ii(x)} \Big)_x \]
 and the Iniquity Theorem holds again.
\end{proof}

\paragraph{The asymmetric case}
In the case of  multiple source-destination pairs 
%, the Iniquity Theorem does not guarantee anymore that the inequality coefficient of the whole population will worsen. However, using the Iniquity Theorem in the symmetric case, 
 the inequality within each set of players in any single commodity is again worsened as a result of tolls. Such a statement is not obtainable for the society as a whole. In Appendix \ref{sec:asymmetric} we show how to create, admittedly contrived, counterexamples  where despite the fact that within each subpopulation the inequality worsens the population as a whole becomes more equal (e.g. the rich and poor use different subnetworks and only the rich get taxed).  We believe that such adversarial counterexamples may be circumvented by imposing more realistic models, and pose this as one of the possible directions for future work.
 %Two formulations of the asymmetric case problem are discussed in Appendix \ref{sec:asymmetric}, with counterexamples showing that neither leads to a general Iniquity Theorem.
 % Intrinsically, this is due to the inherent difficulty of comparing populations that experience different objectives. Specifically it is not clear how to compare a player's path  to another if they do not share the same source and/or destination.

\paragraph{Computing the efficiency-equality trade-off}
In a {\em parallel links network} serving a population with a known income distribution, the routing and tolls that optimize any desired trade-off between efficiency and equality can be computed via dynamic programming. Such an algorithm is given in Appendix \ref{sec:tradeoff}.

%The {\em cumulative income} function is $Q(t)=\int_0^t q(\tau) d\tau$.  The {\em Gini coefficient}  (also known as Gini ratio or Gini index) \cite{?} of the income distribution $\ii(x)$ is $2\int_0^1(x - {\ii(x)\over Q(1)})dx$.

\section{The Iniquity Index}
\label{sec:iniquity}
The Iniquity Theorem shows that under general conditions of the cost functions, the income inequality between agents increases after tolls are levied.
 In this section, we quantify this deterioration of equality by introducing a new metric.
We have captured the importance of the game costs to the agents' income by a parameter \( \alpha > 0\), intuitively small.
 The iniquity (index) is defined as the derivative of the Gini coefficient as \( \alpha \) goes to zero.

\begin{definition}
  Let \( \Gamma \) be a nonatomic symmetric congestion game.
  Agents have an initial \textit{ex ante} distribution \( (\ii(x))_{x \in [0,1]} \) and incur a cost \( \cost^F(x) \) under flow \( F \).
  Let \( \ii_\alpha(x) = \ii(x) - \alpha \cdot \cost^F(x) \) be the \textit{ex post} income distribution for some \( \alpha > 0 \).
   %Let \( G \) be some income inequality measure.
  The \textit{iniquity} of \( \Gamma \) is defined as
  \[
    I(\Gamma) = \lim_{\alpha \to 0^+} \frac{G(q_\alpha) - G(q)}{\alpha}.
  \]
\end{definition}

Note that this notion is well-defined. The Gini coefficient for distribution \( q_\alpha \) is given by
\[
  G(q_\alpha) = 1-2\frac{\int_0^1 \int_0^x (q(t) - \alpha \cdot \cost^F(t)) dt dx}
    {\int_0^1 (q(x) - \alpha \cdot \cost^F(x)) dx} =
  1-2\frac{\int_0^1 Q(x)dx - \alpha \int_0^1 \int_0^x \cost^F(t) dt dx}
    {\mu - \alpha \cdot SC}
\]
where \( \mu \) is the total income of distribution \( q \) and \( SC \) is the social cost. This function is indeed differentiable with respect to \( \alpha \), provided the obvious requirement of \( \mu > 0 \) is satisfied.

The Iniquity Theorem implies that the iniquity is always nonnegative. For the rest of the paper we will focus on the canonical cost functions~\eqref{eq:canonical}.
 As a reminder, the cost of agent \( x \) in edge \( e \)  is
\begin{equation*}
  \ii(x) \cdot f_e(\ii(x), c^F(e), \toll_e) =
    \ii(x) \cdot \latency{e}(c^F(e)) + \toll_e\,.
\end{equation*}

The canonical cost functions, besides having strong experimental justification \citep{abrantes2011meta, borjesson2012income}
provide also significant advantages in the theoretical study of iniquity. Specifically, the iniquity index is invariant under scaling of the population incomes.

\begin{theorem}
  \label{thm:scaleinv}
  {\bf Robustness under scaling of income.} Assume agent cost functions are in canonical form~\eqref{eq:canonical} in a game \( \Gamma \). Then the iniquity is scale invariant: if all incomes are scaled by a constant \( \lambda > 0 \) and optimal tolls are used in the resulting game \( \Gamma_\lambda \), then \( \iniquity(\Gamma) = \iniquity(\Gamma_\lambda) \).
\end{theorem}

\begin{proof}
  Under Equation~\eqref{eq:canonical}, the \textit{edge} cost functions are of the form:
  \[
    f_e(\ii, z, \toll_e) = \frac{\toll_e}{\ii} + \latency{}(z).
  \]
  Optimal tolls \( \toll^*_e \) in the original (unscaled) game \( \Gamma \) minimize the social cost (the sum of all players latencies) as defined by the social planner. In the scaled version of the game, \( \Gamma_\lambda \), the new optimal tolls \( \hat{\toll}^*_e \) should be such that the resulting flow is identical to the minimizing flow in the original game. In that case, the social optimum is realized for \( \Gamma_\lambda \).

  It is possible to show that this result holds for \( \hat{\toll}^*_e = \lambda \toll^*_e \). Indeed, we immediately get
  \[
    f_e(\lambda \ii, z, \lambda \toll^*_e) = \frac{\lambda \toll^*_e}{\lambda \ii} + \latency{}(z)
      = \frac{\toll^*_e}{\ii} + \latency{}(z) = f_e(\ii, z, \toll^*_e),
  \]
  implying that players' costs in the scaled version of the game are identical to the costs in the original game.

  The Gini coefficient is scale invariant in the sense that if all incomes of distribution \( q \) are multiplied by \( \lambda > 0 \), then \( G(q) = G(\lambda q) \). In the canonical form~\eqref{eq:canonical}, the \textit{ex post} distribution is
  \[
    (\epii(x))_x = \Big( \ii(x) (1 - \alpha \cdot \sum_{e\in F(x)} f_e(\ii(x), c^F(e), \toll^*_e)) \Big)_x
  \]
  If incomes are scaled by \( \lambda > 0 \) and optimal tolls are selected by the social planner, the new distribution is
  \[
    \Big( \lambda \ii(x) (1 - \alpha \cdot \sum_{e\in F(x)} f_e(\lambda \ii(x), c^F(e), \lambda \toll^*_e)) \Big)_x =
    \lambda \cdot \Big( \ii(x) (1 - \alpha \cdot \sum_{e\in F(x)} f_e(\ii(x), c^F(e), \toll^*_e)) \Big)_x
  \]
  for which the Gini coefficient is equal to \( G(\epii) \). This further implies that the iniquity of game \( \Gamma_\lambda \) is equal to that of \( \Gamma \).
\end{proof}

% The proof of scale invariance is given in Appendix \ref{sec:proofscaleinv}.

\paragraph{No-regret learning} So far we have looked at the iniquity index in the context of agents playing the Nash Equilibrium of the routing game. However, it is possible to relax this assumption and let agents implement a no-regret strategy of their own.

 Let \( F_1, F_2, \dots \) be a sequence of flows obtained from agents repeatedly playing the game. Agent \( x \) is implementing a no-regret algorithm if it has vanishing regret, i.e.
\[
  R(T) = \frac{1}{T} \sum_{i = 1}^T cost^{F_i}(x) -
    \min_{p \in \mathcal{P}} \frac{1}{T} \sum_{i = 1}^T \sum_{e \in p} f_e(\ii(x), c^{F_i}(e), \toll_e) \to 0 \text{ as } T \to \infty
\]

\newcommand{\appflow}{F_\epsilon}
We also call an \( \epsilon \)-approximate Nash Equilibrium a flow \( F_\epsilon \) such that
\[
  \int_0^1 \cost^{\appflow}(x) dx - \min_{p \in \mathcal{P}} \sum_{e \in p} f_e(q(x), c^{\appflow(e)}, \toll_e) \leq \epsilon \, .
\]

Following the results in \citep{blum2006routing}, we can show that under regret minimizing agents, the flow
converges to that of an approximate equilibrium under the assumption of a finite number of wealth/income levels \( w_1, \dots, w_K \). This assumption is rather realistic since in practice there can only be a finite number of income levels. Also, any continuous distributions over incomes can be approximated to arbitrary high accuracy by a distribution of finite but large enough support.

\begin{theorem}
  \label{thm:noregret}
 {\bf  Robustness under no-regret learning.} Given a finite number of income levels, the iniquity index is uniquely defined under the assumption of no-regret learning agents. Specifically, if all agents follow a no-regret algorithm, we have
  \[
    \lim_{\alpha \to 0; \alpha > 0} \lim_{T \to \infty} \frac{\frac{1}{T}\sum_{t = 1}^T G(\epii^t) - G(\ii)}{\alpha} =
      I(\Gamma)
  \]
where \( \epii^t \) is the \textit{ex post} income distribution of the \( t \)-th instance of the game.
\end{theorem}

\begin{proof}
The proof consists of two steps. In the first step we will show that our symmetric game of type-specific costs \( \Gamma \) reduces to an asymmetric congestion game \( \hat{\Gamma} \). In the second step, we will apply results about the behavior of no-regret dynamics in asymmetric congestion games from \citep{blum2006routing} to prove the robustness of the iniquity index.

For every edge \( e \) in \( \Gamma \), construct the parallel edges \( (\hat{e}_i)_{i=1}^K \) linking \( e \) to its original endpoint.
The cost of edge \( \hat{e}_i \) is constant and equal to \( \frac{\toll_e}{w_i} \).
Now for each path \( p \in \mathcal{P} \), the player of type \( i \) has an associated path \( \hat{p} \in \hat{\mathcal{P}} \) that uses all the edges in \( p \) as well as \( \hat{e}_i \). Since \( \Gamma \) and \( \hat{\Gamma} \) are payoff-equivalent games, under the assumption that our latency functions have bounded slope we can use the results from \citep{blum2006routing} to show that the iniquity is stable under no-regret learning.

Under assumptions \eqref{eq:cf1} and \eqref{eq:cf2} on the form of the agent cost functions, as long as latency functions are of bounded slope, \( \cost^F(x) \) is of bounded slope, since the income of players and tolls are bounded. No-regret algorithms will therefore approach an approximate \( \epsilon \)-Nash equilibrium for some \( \epsilon > 0 \). Lastly, we will show that this implies the iniquity of the game \( \Gamma \) under no-regret algorithms approaches that of \( \Gamma \).

Let \( \alpha \) and some \( \epsilon \) be fixed. There exists a time \( T_\epsilon \) such that after \( T_\epsilon \), \( R(T) \leq \epsilon \).
By \citep{blum2006routing}, at most a fraction \( K \epsilon^\frac{1}{4} \) of the first \( T_\epsilon \) games is not an \( f(\epsilon) \)-NE, where \( f \) is a function that goes to zero when \( \epsilon \) also goes to zero
and \( K \) is a constant.\footnote{Precisely, there is a fraction \( ms^{\frac{1}{4}}\epsilon^{\frac{1}{4}} \) of time periods where the flow is not an \( \epsilon + 2 \sqrt{s\epsilon n} + 2m^\frac{3}{4} s^\frac{1}{4} \epsilon^\frac{1}{4} \)-NE, where \( m \) is the number of edges, \( s \) is a bound on the slope of edge cost functions and \( n \) is the largest number of edges in a path.}
\( T_\epsilon \) is a function in \( O(\frac{1}{\epsilon^2}) \) so as \( \epsilon \) goes to zero we have that \( T_\epsilon \) goes to \( \infty \).
Call \( A_\epsilon \) the set of time periods where \( \epii^t \) is obtained from an \( f(\epsilon) \)-approximate NE and \( B_\epsilon \) where this is not the case. We have
\begin{align}
  & \frac{1}{T_\epsilon} \sum_{t = 1}^{T_\epsilon} G(\epii^t) =
    \frac{1}{T_\epsilon} \sum_{t \in A_\epsilon} G(\epii^t) + \frac{1}{T_\epsilon} \sum_{t \in B_\epsilon} G(\epii^t) \label{eq:noregret} \\
    \text{and } & \, \frac{1}{T_\epsilon} \sum_{t \in B_\epsilon} G(\epii^t) \leq \frac{1}{T_\epsilon} K \epsilon^\frac{1}{4} \cdot T_\epsilon = K \epsilon^\frac{1}{4} \label{eq:beps}
\end{align}
where the last inequality holds due to the Gini coefficient being bounded by 0 and 1. Thus, the maximum distance between the Gini coefficient of \( \epii^t \) and that of the NE is 1, for the at most \( \epsilon T_\epsilon \) time steps where we do not have an \( f(\epsilon) \)-NE.

It remains to prove that for the other time steps in \( A_\epsilon \), we are approaching the NE costs that give rise to \( \epii \) as \( \epsilon \) goes to zero.
Indeed, let \( \appflow \) be the flow corresponding to an \( \epsilon \)-approximate NE.
In congestion games, the vector of path costs \( (C_p)_{p \in \mathcal{P}} \) realized at a NE is unique.
Take a decreasing sequence of \( (\epsilon_k)_k \to 0 \) as \( k \to \infty \) and flows \( F_{\epsilon_k} \) that are \( \epsilon_k \)-NE.
Their associated cost vectors are \( (C_p^{\epsilon_k})_p \). Suppose that as \( \lim_{k \to \infty} (C_p^{\epsilon_k})_p \neq (C_p)_p \).
By compactness, up to a subsequence, \( F_{\epsilon_k} \) converges to some flow \( F \) that is a NE. But for this flow \( F \) the path costs are different from \( (C_p)_p \). This is a contradiction.

On the other hand, the Gini coefficient is a continuous function of the agents' costs, so any sequence of \( \epsilon \)-NE approximating a NE of the game will approach its Gini coefficient. Coming back to Equations~\eqref{eq:noregret} and~\eqref{eq:beps}, we know that
\[
  \frac{1}{T_\epsilon} \sum_{t = 1}^{T_\epsilon} G(\epii^t) \sim_{\epsilon \to 0}
    (1-\epsilon) G(\epii) + O(\epsilon^\frac{1}{4})
\]
and thus \( \frac{1}{T} \sum_{t = 1}^{T} G(\epii^t) \to G(\epii) \) as \( T \to \infty \).
\end{proof}

\section{Computing the Iniquity in Pigou}
\label{sec:pigou}
To illustrate the interplay between wealth or income and congestion
games, we consider the well-studied Pigou network, which consists of
two parallel links (Figure~\ref{fig:pigou}) with latency functions
$\latency{u}(r)=1$ and $\latency{d}(r)=r$. Assume that this
transportation network is used by a population of (normalized) size 1
and with wealth or income function $\ii(x)=x^\beta$, for some
nonnegative parameter~$\beta$.  We assume that the population is
ordered by income and normalized, so that $x$ represents the quantile
of members that are richer than an $x$ fraction of the total
population.

The perceived cost for quantile $x$ is the canonical form~\eqref{eq:canonical}
$\cost(x)=\latency{e}(\flow{e}) \cdot \ii(x)+\toll_e$, where $e=e(x)$ is the edge used by
the quantile, $\flow{e}$ is the flow through link $e$, and $\toll_e$ the
toll on link $e$.  It is not hard to argue that at equilibrium the
$\s$ fraction of the population that uses the upper link is the
poorest $\s$ part of the population, and therefore the perceived cost
is given by
\begin{align}
  \label{eq:x1}
  \cost(x)=
  \begin{cases}
    \latency{u}(\s) \ii(x)+\toll_u & x\leq \s \\
    \latency{d}(1-\s) \ii(x) + \toll_d & \text{otherwise}
  \end{cases}
                               &=\begin{cases}
                                 \ii(x)+\toll_u & x\leq \s \\
                                 (1-\s) \ii(x) + \toll_d &
                                 \text{otherwise}
                               \end{cases}
\end{align}
By continuity, \emph{at equilibrium the perceived cost of quantile $\s$ must be
the same in both links}, from which we get
$\ii(\s)+\toll_u=(1-\s) \ii(\s)+\toll_d$, or equivalently that the
difference in tolls $\toll=\toll_d-\toll_u$ is
\begin{align}
  \label{eq:x2}
  \toll=\ii(\s) \s.
\end{align}
Since this is always positive, without loss of generality we will
assume that there is no toll on the upper link and toll
$\toll=\ii(\s) \s$ in the lower link. We want to investigate the
effects of toll $\toll$ on the Gini coefficient, but for simplicity we
often use $\s$ as primary variable instead of $\toll$ (as suggested by
Equation~\eqref{eq:x2}; since $\s$ is monotone, we can easily
translate expressions with $\s$ to expressions with $\toll$).

% Let's assume that the tolls are selected to minimize the social cost
% \begin{align*}
%   \int_0^1 \cost(x)\, dx &=\int_0^1
%                            \begin{cases} \ii(x) & x \leq \s \\
%                              (1-\s)\ii(x) + \toll & \text{otherwise}
%                            \end{cases}
%                                                     =\int_0^1
%                                                     \begin{cases}
%                                                       \ii(x) & x \leq \s \\
%                                                       (1-\s)\ii(x) +
%                                                       \s \ii(\s) &
%                                                       \text{otherwise}
%                                                     \end{cases} \\
%                          &= \int_0^{\s} \ii(x) \, dx+\int_{\s}^1
%                            (1-\s)\ii(x) + \s
%                            \ii(\s) \, dx \\
%                          &= \int_0^1 \ii(x) \, dx-\int_{\s}^1
%                            \s(\ii(x) -
%                            \ii(\s)) \, dx \\
% \end{align*}

To derive simpler analytical results in the following, we multiply the income distribution by a constant equal to \( \beta + 1 \). This however does not change the nature of the analysis due to Proposition \ref{thm:scaleinv} stating that under the canonical form~\eqref{eq:canonical} and optimal tolls, scaling the income distribution by a constant yields the same iniquity.

For the income distribution $\ii(x)=2x$, and the latencies $\latency{u}(x)=1$
and $\latency{d}(x)=x$, the toll is $\toll=2\s^2$
and \( \cost(x) = 2x \) if \( x\leq \s \), \( \cost(x) = (1-\s)2x + 2\s^2 \) otherwise.
Figure~\ref{fig:Pigou-social-cost} shows the social cost
$\int_0^1 \cost(x)\, dx=1 - \s + 2 \s^2 - \s^3$ as a function of
the toll $\toll$ and as a function of the switching point $\s$.  The
minimum value $23/27$ is achieved when the toll is $\toll=2/9$ and
$\s=1/3$ of the population uses the upper link. Note that this
minimum value is less than $1$, the best social cost with tolls when
we don't take into account the income distribution (or equivalently
in a society with Gini coefficient 0).

The Iniquity Theorem states that the shape~\eqref{eq:cf1} can be used for quantifying the agents' costs, effectively the canonical form divided by the income of the agent. This indeed gives similar results that are presented in Appendix \ref{sec:pigoucf1}.
% If instead we focus on the perceived latency $\cost(x)/2x$ (actual
% latency plus tolls over income) as per Equation~\eqref{eq:cf1}, we get the associated social
% perceived latency
% $\int_0^1 \cost(x)/2x\, dx=1 - \s+ \s^2 - \s^2 \ln \s$, shown in
% Figure~\ref{fig:Pigou-social-latency}. The perceived latency is
% similar to social cost, but distorted because we divide the social
% cost of each participant by their income.

Finally, for comparison with
the situation when income is not taken into account (or equivalently,
in a society with no inequality),
Figure~\ref{fig:Pigou-actual-latency} shows the corresponding actual
latency (time wasted in transportation). It depicts the social actual
latency $\int_0^{\s}1\,dx+\int_{\s}^1 1-\s\,dx=\s+(1-\s)^2$. The right
part of the figure is independent of the income distribution, and
shows the familiar minimum value $3/4$ (corresponding to price of
anarchy $4/3$), achieved when the flow is split equally. The left part
of figure shows the same quantity in terms of the toll $\toll$; this
figure depends on the income distribution because the toll
$\toll=\ii(\s) \s$ depends on the income distribution.

\begin{figure}
  \centering
  \begin{minipage}[c]{0.35\textwidth}
		\centering
		\includegraphics[width=\textwidth]{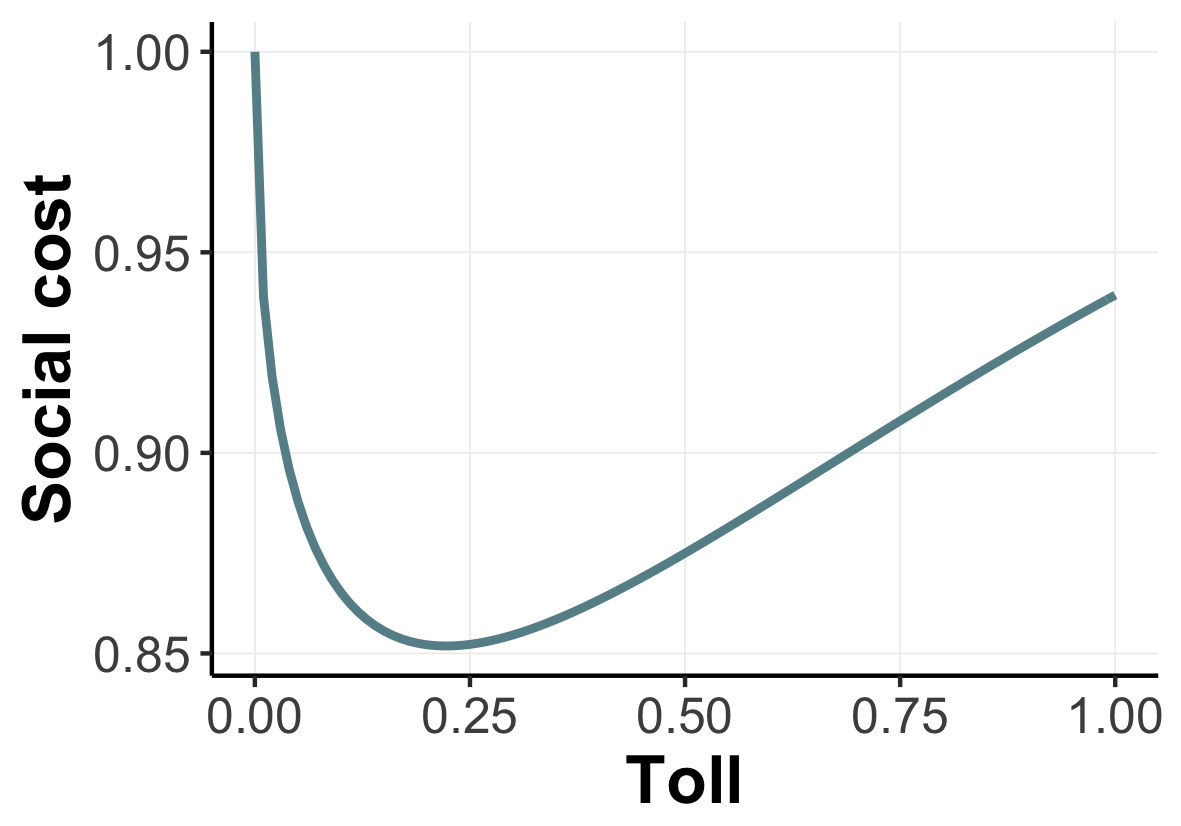}
  \end{minipage}
  \begin{minipage}[c]{0.15\textwidth}
		\begin{tikzpicture}
		\end{tikzpicture}
	\end{minipage}
  \begin{minipage}[c]{0.35\textwidth}
		\centering
		\includegraphics[width=\textwidth]{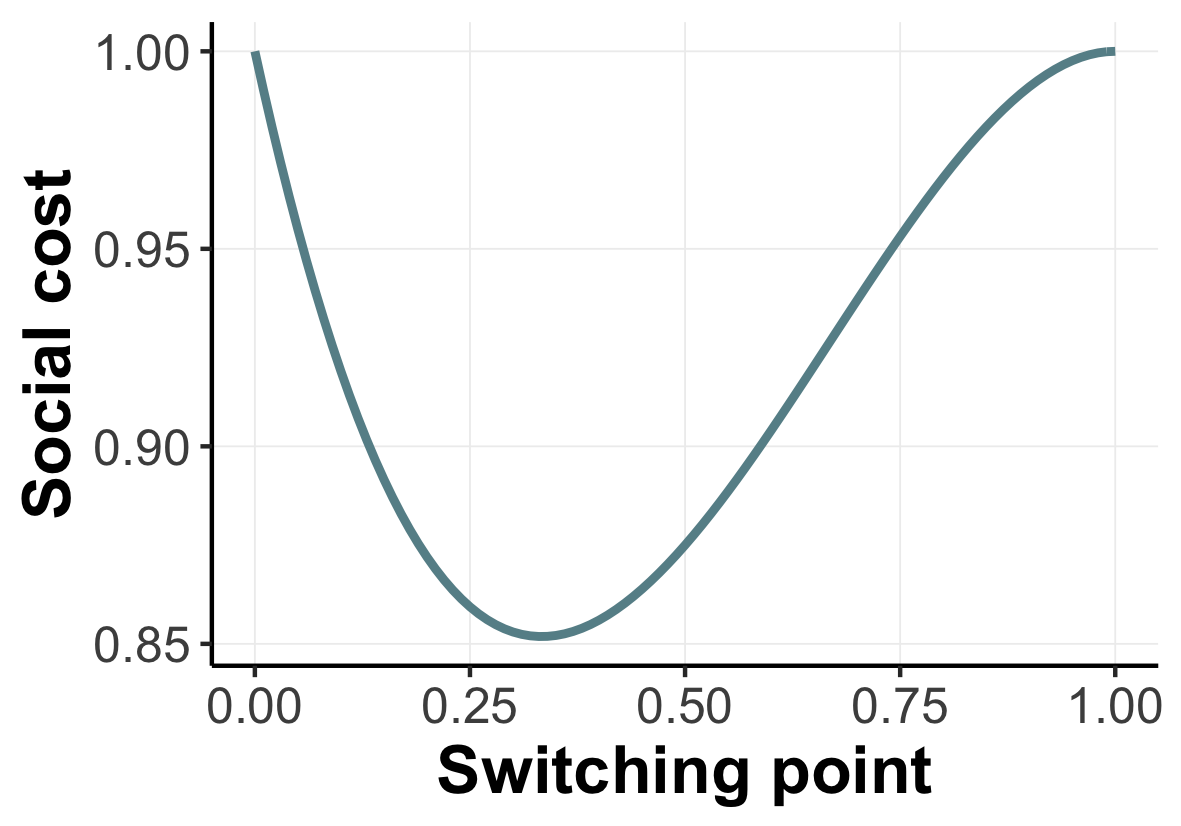}
  \end{minipage}
  \caption{The social cost (income loss) of the Pigou network with
    income distribution $\ii(x)=2x$, as a function of the tolls (left)
    and switching point $\s$ (right). The minimum value $23/27$ is
    achieved at $\toll=2/9$ and $\s=1/3$. Note that this minimum value
    is less than $1$, the best social cost with tolls when we don't
    take into account the income distributions (or equivalently in a
    society with Gini coefficient 0).}
  \label{fig:Pigou-social-cost}
\end{figure}

% \begin{figure}
%   \centering
%   \includegraphics[scale=0.5]{Figures/Pigou-social-perceived-latency-tau}
%   \includegraphics[scale=0.5]{Figures/Pigou-social-perceived-latency-fu}
%   \caption{The perceived latency (time wasted in transportation plus
%     tolls over income) of the Pigou network with income
%     distribution $\ii(x)=2x$, as a function of the tolls (left) and switching
%     point $\s$ (right).
%     % The minimum value $23/27$ achieved at $\toll=2/9$ and $\s=1/3$.
%     % Note that the minimum value is less than $1$, the best social cost
%     % with tolls when we don't take into account the income distributions
%     % (or equivalently in a society with Gini coefficient 0).
%   }
%   \label{fig:Pigou-social-latency}
% \end{figure}

\begin{figure}
  \centering
  \begin{minipage}[c]{0.35\textwidth}
		\centering
		\includegraphics[width=\textwidth]{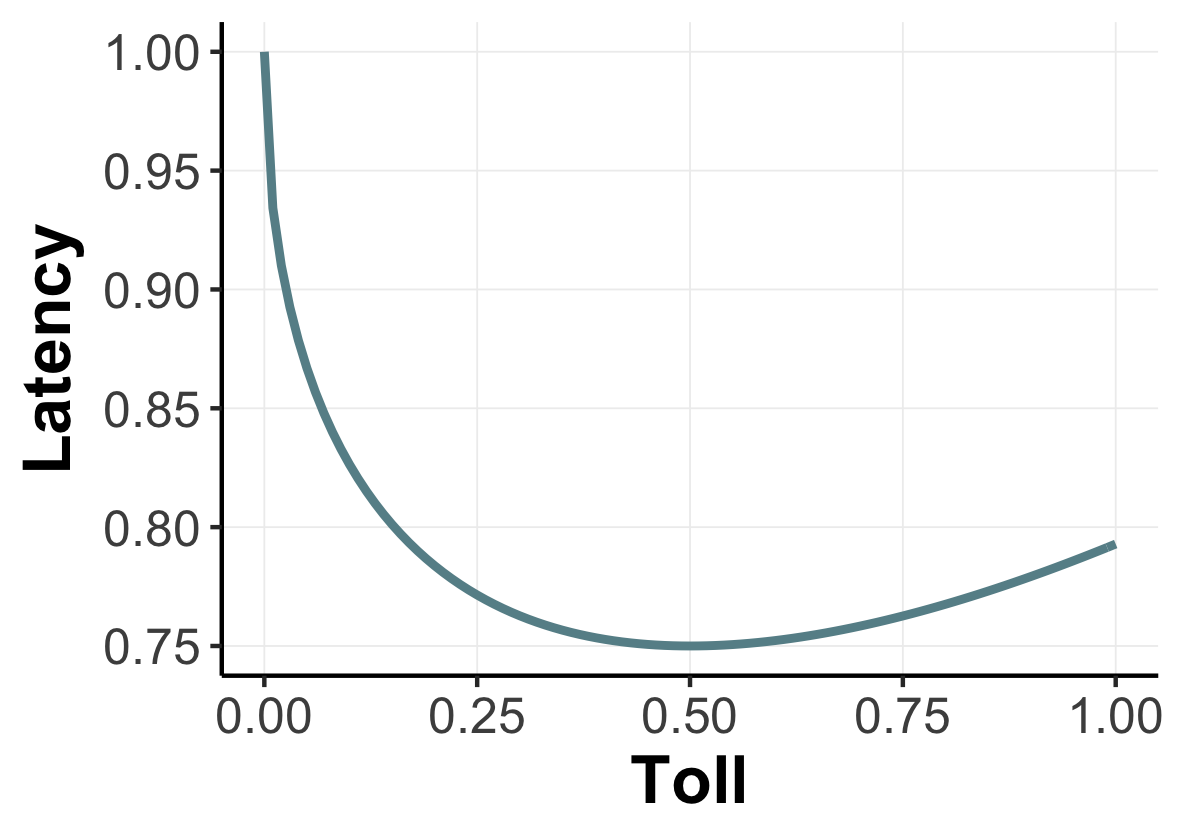}
  \end{minipage}
  \begin{minipage}[c]{0.15\textwidth}
		\begin{tikzpicture}
		\end{tikzpicture}
	\end{minipage}
  \begin{minipage}[c]{0.35\textwidth}
		\centering
		\includegraphics[width=\textwidth]{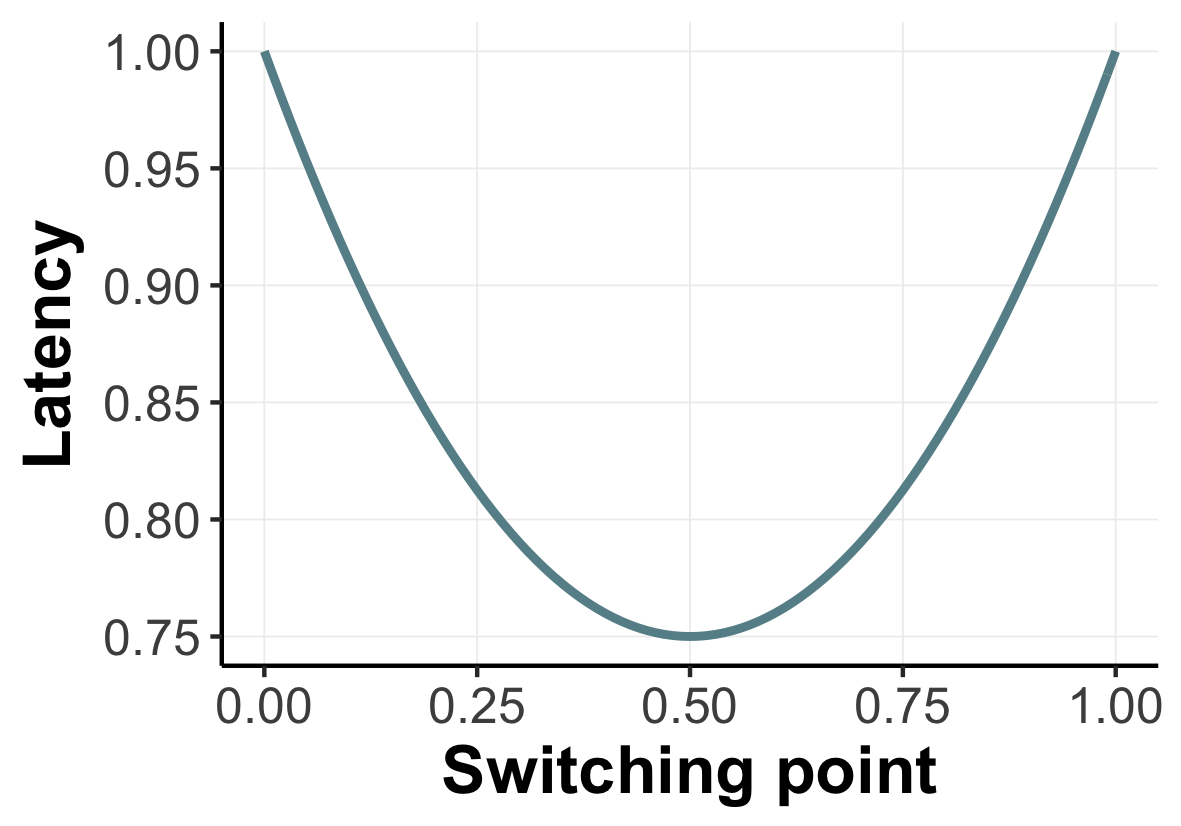}
  \end{minipage}
  \caption{The actual latency (time wasted in transportation) of the
    Pigou network as a function of the tolls (left) and switching
    point $\s$ (right). The right figure is independent of the income
    distribution. The minimum value $3/4$ is achieved at $\s=1/2$ and
    corresponds to the well-known price of anarchy $4/3$. The left
    figure depends on the income distribution because the toll
    $\toll=\ii(\s)\s$ depends on the income distribution. Here the
    income distribution is $\ii(x)=2x$.}
  \label{fig:Pigou-actual-latency}
\end{figure}

% The Gini coefficient of a population with un-normalized quantile
% income $\ii(x)$ is given by
% \begin{align*}
%   \Gini=1-2\frac{\int_0^1 \int_0^x
%     \ii(t)\,dt\,dx}{\int_0^1 \ii(x) \, dx}.
% \end{align*}
% The denominator essentially normalizes the income so that the total
% income is 1. It is important to keep in mind that even when we start
% with a normalized income distribution, the denominator may change as
% the income distribution changes when we take into account the
% transportation latency and tolls.

Let $\epii(x)=\ii(x)-\alpha\cdot\cost(x)$ be the perceived income when
we take into account the effects of perceived latency into the actual
income, where $\alpha$ indicates the importance of transportation. We are
interested in first order effects, so we will always assume that
$\alpha$ is very small and that in fact it tends to 0. Let's define as
$\Gini(\alpha, \toll)$ the inequality coefficient when we take into account
the effects on the income of the transportation cost, assuming that
the social designer selects toll $\toll$.

We will assume that \emph{the social designer selects tolls to minimize the
actual latency} on the network (the minimum in
Figure~\ref{fig:Pigou-actual-latency}).\footnote{There are reasonable
alternatives for the social planner, such as minimizing the social
cost, that we don't explore in this work.} For the Pigou network
the optimal switching point is $\s=1/2$.  For income distribution
$\ii(x)=(\beta+1)x^{\beta}$, this optimal switching point corresponds
to toll $\toll=\ii(\s)\s=(\beta+1)2^{-(\beta+1)}$. It is immediate
to see that for these income distributions, Equation~\eqref{eq:x1} for
the cost becomes \( \cost(x)= (\beta+1)x^{\beta} \) for \( x \leq \s \) and \( \cost(x) = (1-\s)(\beta+1)x^{\beta} + \toll \) otherwise.

We can now compute directly the Gini coefficient. The trick is
to express everything in terms of the switching point $\s$, instead of
the toll $\toll$. Since at equilibrium, Equation~\eqref{eq:x2},
$\toll=\ii(\s) \s=(\beta+1)\s^{\beta+1}=(\beta+1)2^{-(\beta+1)}$, we
have
\begin{align*}
\Gini(\alpha, (\beta+1)2^{-(\beta+1)} )
  % &= 1-2\frac{\int_0^1 \int_0^x \epii(t)\,dt\,dx}{\int_0^1 \epii(x) \, dx} \\
  &= \frac{\beta \left(\alpha \left(\beta +2^{\beta +2}+3\right)-2^{\beta
    +3}\right)}{2 (\beta +2) \left(\alpha \left(\beta +2^{\beta+1}+2\right)-2^{\beta +2}\right)} =
  \frac{\beta }{\beta +2}+\frac{\beta  (\beta
    +1)}{(\beta +2) 2^{\beta +3} } \alpha  + O\left(\alpha ^2\right).
 \end{align*}
 The last expression comes from the Maclaurin expansion of the
 function, from which we derive the following Theorem.
 \begin{theorem}
     \label{thm:pigou}
   For the Pigou network with two links and latency functions $1$ and
   $x$, and for a population with income distribution
   $\ii(x)=x^{\beta}$, when tolls are selected to minimize
   the actual latency, the toll at Nash (Wardrop) equilibrium is
   \( \toll=\ii(\s) \s= 2^{-(\beta+1)} \),
   and the Iniquity index is
   \( I(\Gamma) = \frac{d\Gini(0)}{d\alpha}=\frac{\beta  (\beta
     +1)}{(\beta +2) 2^{\beta +3} } \).
\end{theorem}

The values of the iniquity as a function of the income coefficient
$\beta$ are shown in Figure~\ref{fig:iniquity}. The maximum occurs when the income
coefficient $\beta$ is close to 2 (actually when $\beta\approx
1.688$), which means that real-life income distributions have almost
the maximum Iniquity index.

\begin{figure}
  \centering
  \begin{minipage}[c]{0.35\textwidth}
		\centering
		\includegraphics[width=\textwidth]{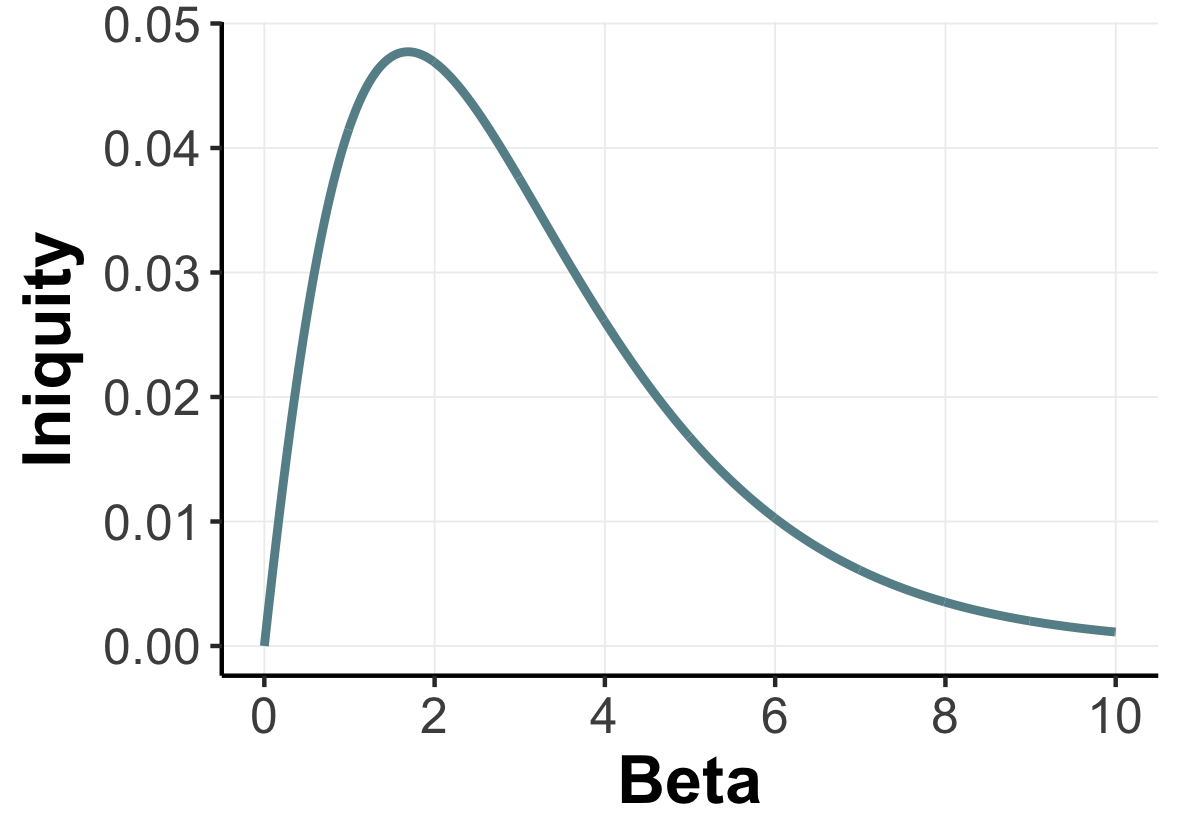}
	\end{minipage}
  \begin{minipage}[c]{0.15\textwidth}
		\begin{tikzpicture}
		\end{tikzpicture}
	\end{minipage}
	\begin{minipage}[c]{0.35\textwidth}
		\centering
    \renewcommand{\arraystretch}{1.4}
    $$
      \begin{array}{|l||cccccc|}
        \hline
        \beta & 0 & 1 & 2 & 3 & 4 & 5 \\
        \hline
        \Gini & 0 & \frac{1}{24} & \frac{3}{64} & \frac{3}{80} &
          \frac{5}{192} & \frac{15}{896} \\
        \hline
      \end{array}
    $$
	\end{minipage}
  \caption{\textit{Left:} The Iniquity index as a function of the income coefficient
    $\beta$. It is 0 when there is no inequality ($\beta=0$) because
    tolls have the same effect on everybody, and rises as the inequality
    increases. At some point, $\beta\approx 1.688$, the Iniquity index
    starts decreasing, because the toll $(\beta+1)2^{-(\beta+1)}$
    becomes small and has little effect on the inequality index. \textit{Right:} Values of the iniquity for different \( \beta \).}
  \label{fig:iniquity}
\end{figure}

\section{Tolls and Inequality: Empirical Findings}
\label{sec:data}

We use detailed transportation data gathered through Singapore's National Science Experiment (NSE) to test how income inequality affects the distribution of transportation delays in a representative sample of students \citep{Monnot2016,monnot2017routing,wilhelm2016sensg}. Although Singapore is the third most densely populated country in the world, the modern infrastructure, cost of private cars, and significant tolls in Singapore minimize congestion on the roads. We examine whether this gain in efficiency incurs costs in terms of income inequality, as predicted by the theoretical results in this paper. The NSE dataset enables us to accurately split student trips in the morning---the time of the day when tolls are most onerous---by the transportation mode (bus, car, walk, and train) \citep{Wilhelm2017}. We then combine the travel data with a dataset on property prices to assess the relationship between income and the average duration and average distance of trips by transportation mode (for details, please see the Appendix \ref{sec:data-app}).

\begin{figure}[!ht]
    \centering
    \begin{minipage}[c]{0.46\linewidth}
        \centering
        \includegraphics[width=0.95\textwidth]{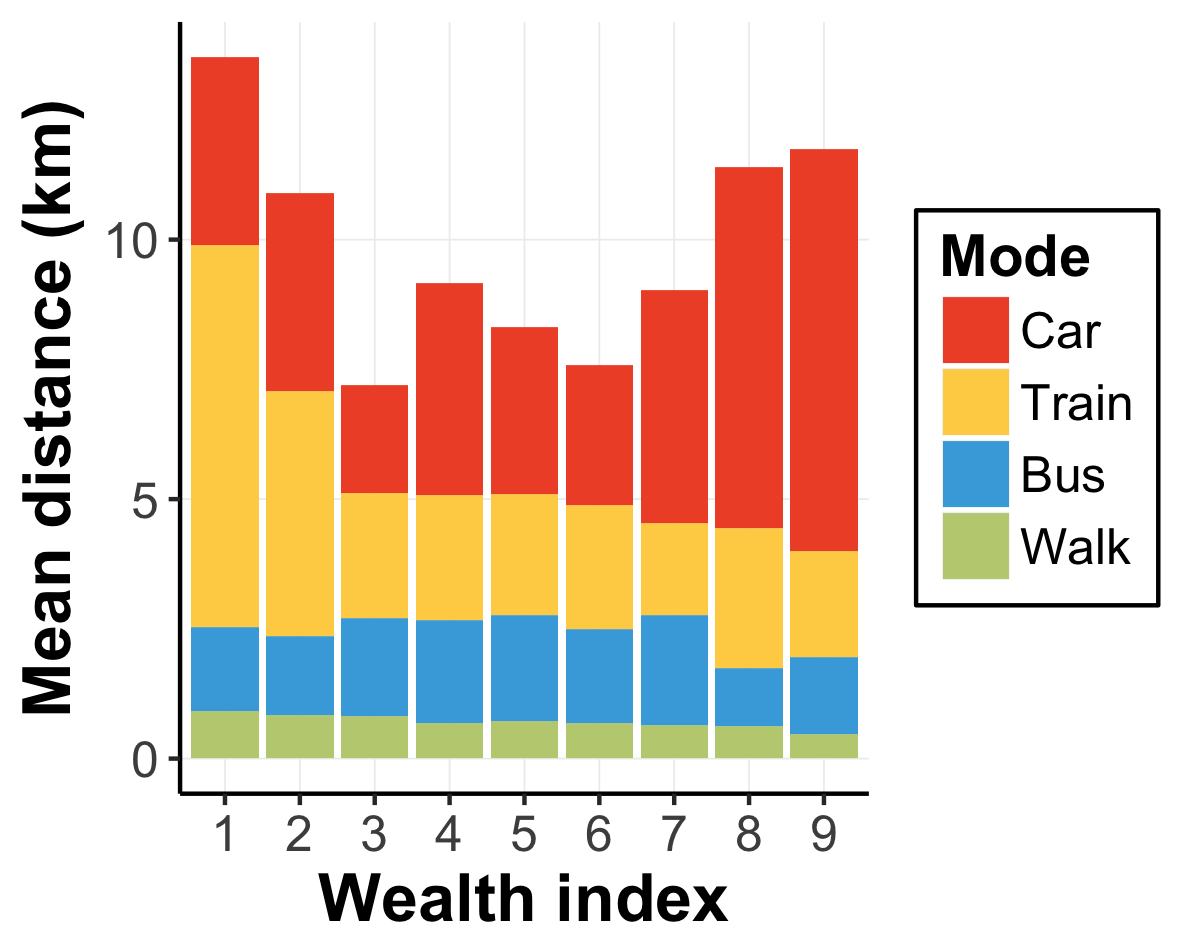}
    \end{minipage}
    \begin{minipage}[c]{0.46\linewidth}
        \centering
        \includegraphics[width=0.95\textwidth]{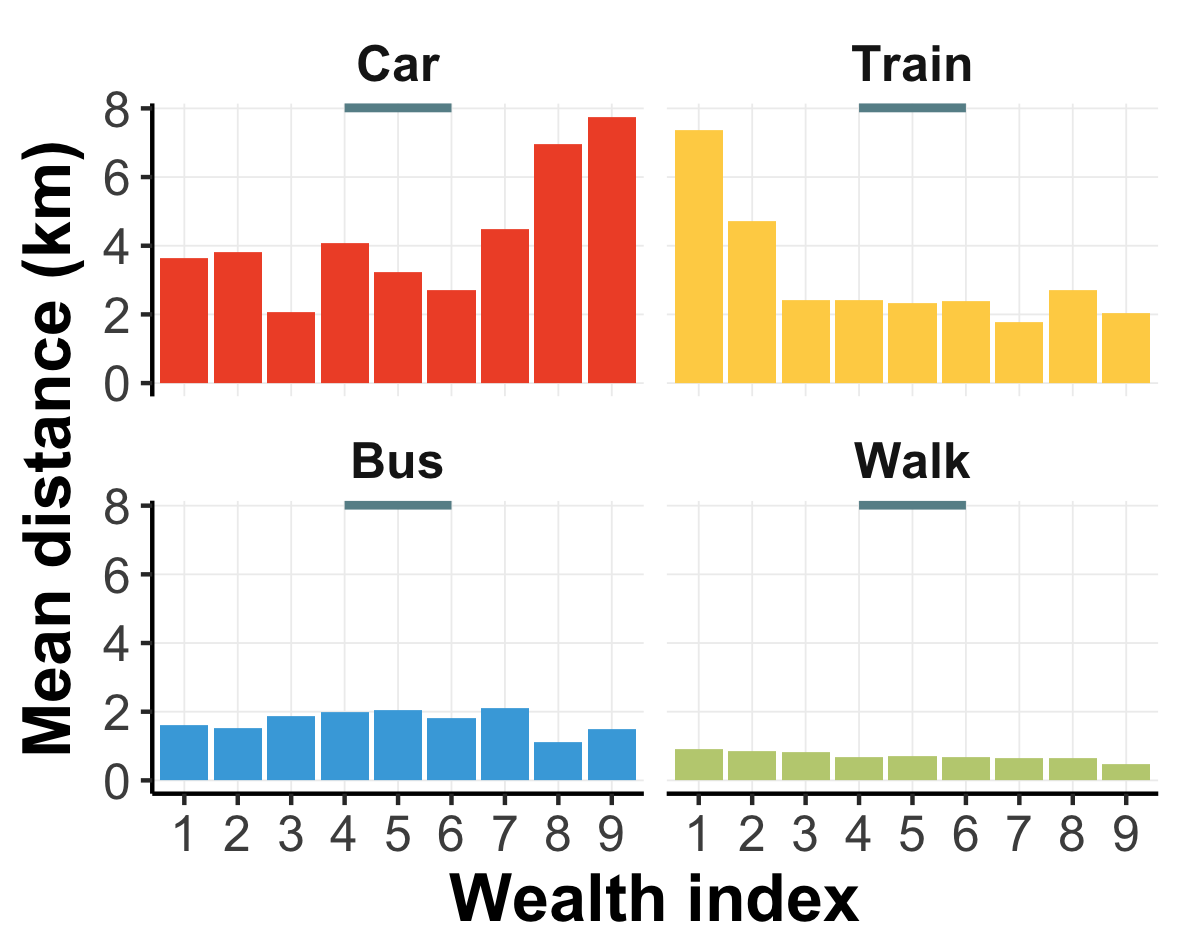}
    \end{minipage}
    \caption{The average trip distance per wealth bracket is presented in the two figures above. \textit{Left:} Average distance in each transport mode per wealth bracket. \textit{Right:} The left plot is split along the transport mode, to show the relative distances logged in each mode across wealth brackets.}
    \label{fig:avgdistpersgm}
\end{figure}

By relying on the sociological literature pertaining to income inequality and Singapore's urban development, we divide the students into 9 wealth brackets based on residence. We then conservatively classify these brackets as low-income, middle-income, and high-income groups \citep{Abeysinghe_AggregateConsumptionPuzzle_2004,Chua_CommunitarianIdeologyDemocracy_1995,trocki_singapore_2006,huff_developmental_1995,Edelstein_HousePricesWealth_2004,Han_PolycentricUrbanDevelopment_2005,Reardon_IncomeInequalityIncome_2011}.
The three groups' homes and schools are plotted in Figure \ref{fig:schperwealth}.
The differences in the means of trip distance and duration are statistically significant among these three groups; and, they lend strong support to the predictions of the Iniquity Theorem. As can be seen in Figure \ref{fig:avgdistpersgm}, when one compares low-income and high-income groups, there is a notable increase in car usage and decrease in the use of walking and public transportation. Because cars are much faster than using bus and walking, the use of cars is associated with a sizable difference in the average duration that students spend in traveling to school (Figure \ref{fig:avgdurpersgm}).

Although students from high-income groups travel a longer distance compared to middle-income groups, this translates into minor differences in travel duration. The opposite is the case when we compare low-income and middle-income groups. These students experience on average 7 to 5 minutes delay compared to middle-income groups, despite the fact that the distance they travel is roughly comparable to high-income groups (Figures \ref{fig:avgdistpersgm} and \ref{fig:avgdurpersgm}). Thus, the Singaporean case---which is an ideal setting to examine the relationship between inequality and transportation delays---offers positive evidence on the Iniquity Theorem. It also provides some lessons on the policies that can be implemented to mitigate the trade-off between efficiency and inequality, as we discuss in the Appendix \ref{sec:data-app}.

\begin{figure}[!ht]
	\centering
	\includegraphics[width=\textwidth]{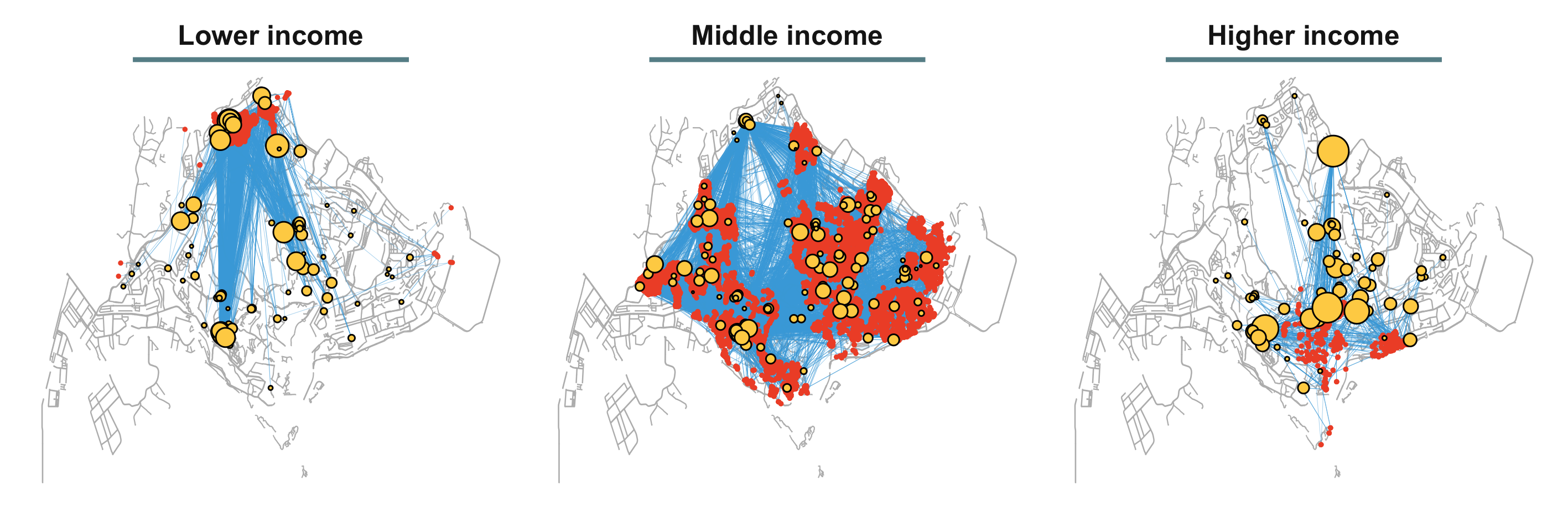}
	\caption{Students are divided in three income groups --- low, middle and high --- using the wealth index obtained from property prices. On the maps of Singapore presented above, the homes of the students are plotted in red and the schools in yellow. The size of the school point reflects the fraction of students attending this particular school. Finally trips are plotted as blue straight lines between the homes and the schools. Further details on the relation between income levels and experienced cost is given in Appendix \ref{sec:data-app}.}
	\label{fig:schperwealth}
\end{figure}

\begin{figure}[!ht]
    \centering
    \begin{minipage}[c]{0.46\linewidth}
        \centering
        \includegraphics[width=0.95\textwidth]{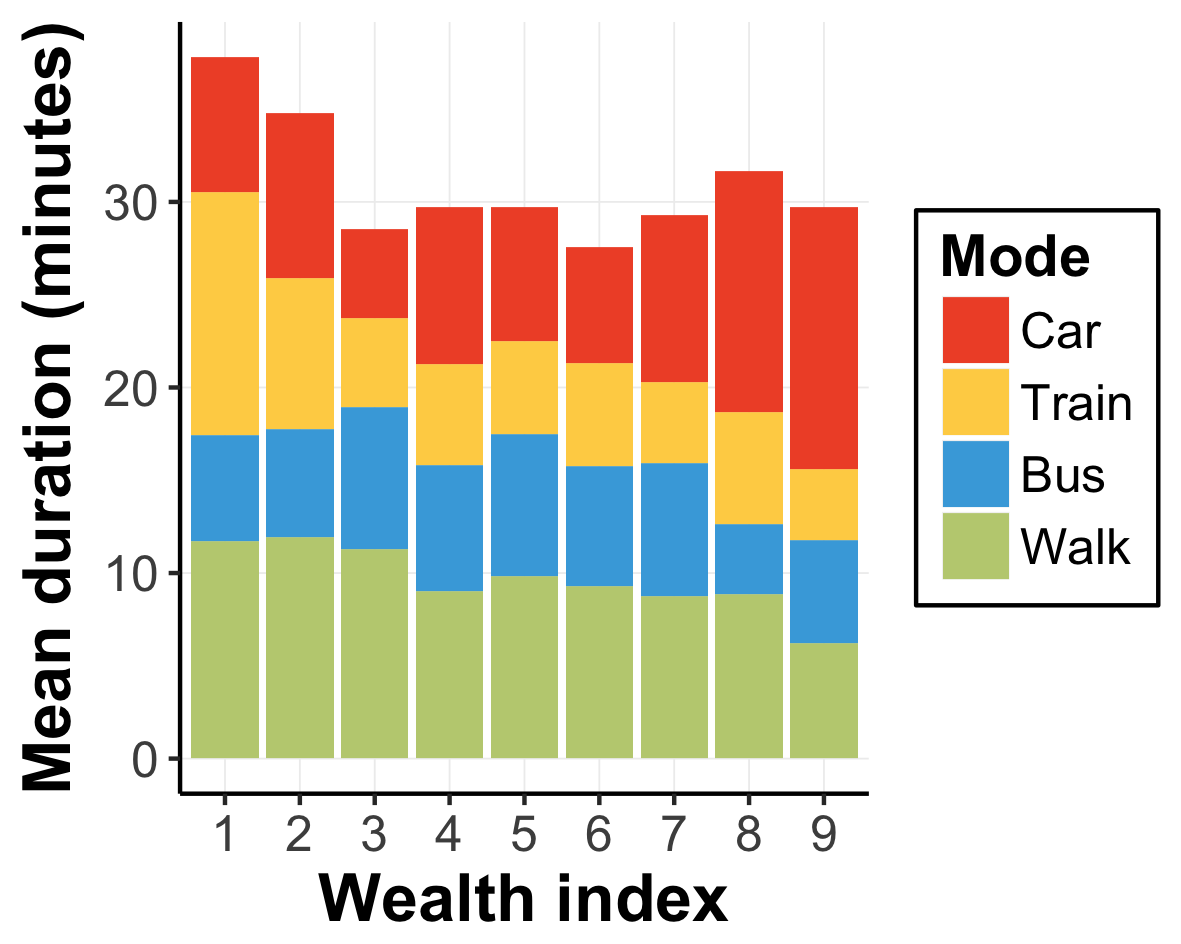}
    \end{minipage}
    \begin{minipage}[c]{0.46\linewidth}
        \centering
        \includegraphics[width=0.95\textwidth]{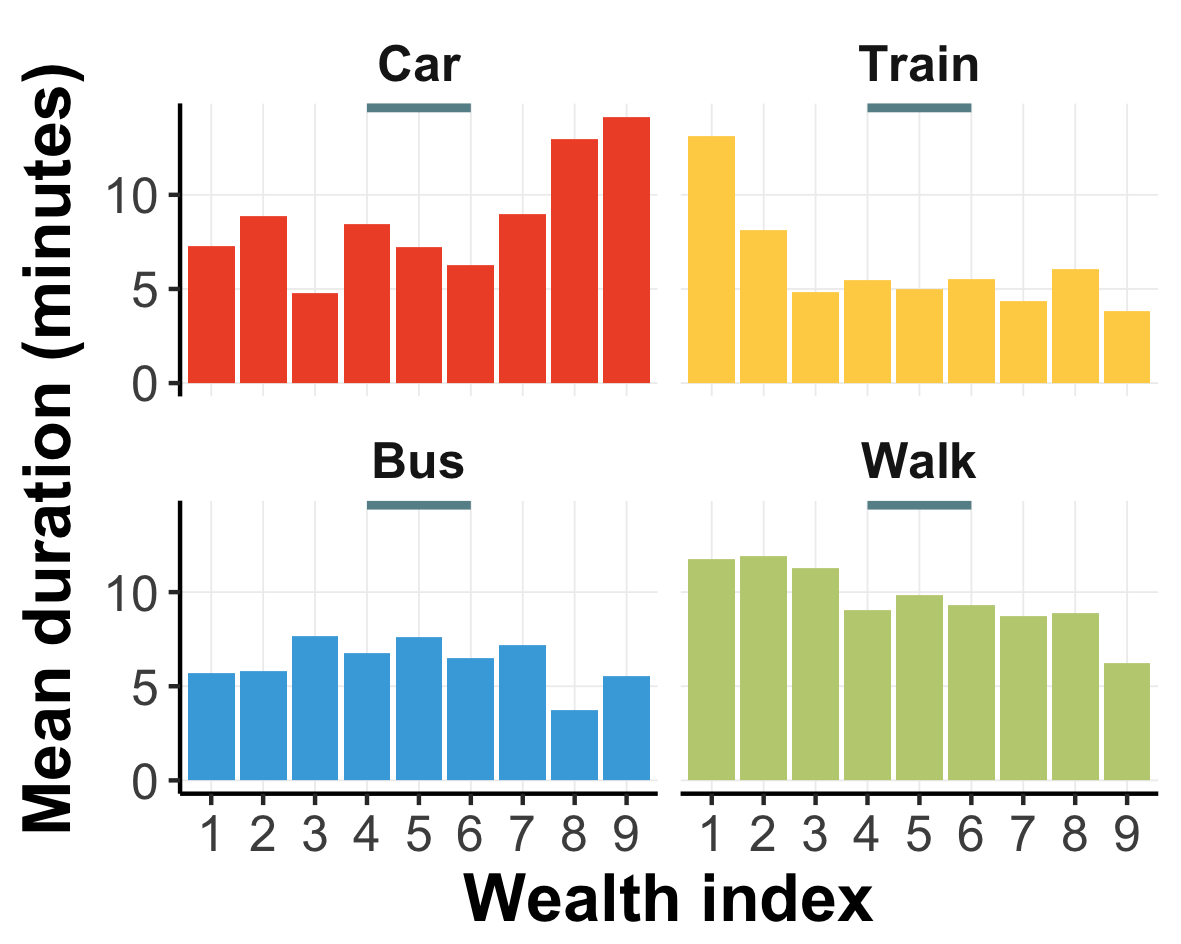}
    \end{minipage}
    \caption{The average trip duration per wealth bracket is presented in the two figures above. \textit{Left:} Average duration in each transport mode. It is notable that brackets 1 and 2 have respectively 7 and 5 minutes more travel time on average than the other brackets. \textit{Right:} The left plot is split along the transport mode, to show the relative durations spent in each mode across wealth brackets. Given that the least-affluent groups spend much more time on the roads compared to middle-income groups, there is a quasi-monotonic increase in the use of car as wealth increases, while we observe that the uses of walking and public transportation decrease as wealth increases.}
    \label{fig:avgdurpersgm}
\end{figure}

\section{Efficiency and Equality: Ramifications for Distributive Justice and Related Questions}
\label{distributive-justice}

The Iniquity Theorem raises important questions pertaining to distributive justice and the efficiency of decentralized decision-making mechanisms such as markets \citep{RawlsJusticeFairnessRestatement2001,RawlsTheoryJustice1999,jencks_does_2002,beckert_market_2002, polanyi_great_1957, gemici_neoclassical_2015}. Namely, if efficiency can be obtained through purposeful intervention but only at a price of increase in inequality, what are the implications of this trade-off for the organization of markets, industries, and society in general? Through a dialogue with contemporary social sciences, we engage only a subset of the rich array of questions raised by the Iniquity Theorem. We thus highlight three potential issues as venues for further research (for a full discussion of these issues, please see the Appendix \ref{sec:justice-app}).

\paragraph{The opportunity cost of iniquity}
One of the fundamental concepts in theories of distributive justice is the idea of fair equality of opportunity, which is the notion that ascriptive factors such as race, socio-economic background, or wealth should not play a decisive role in determining who fails or succeeds in markets and society. A pragmatic justification for this notion comes from the fact that the absence of equal opportunity has costs in terms of economic growth and innovation. The theoretical results in this paper provides a starting point to analyze and compute the costs of iniquity in terms of growth and innovation, topics with important connections to contemporary economics \citep{Benhabib_TradeoffInequalityGrowth_2003,Aghion_InnovationTopIncome_2015,Cingano_TrendsIncomeInequality_2014}.

\paragraph{How does iniquity affect cooperation among members of society?}
There is experimental evidence and a growing body of theoretical research showing that perceptions of fairness have an effect on the likelihood of cooperation among individuals \citep{Colasante_ImpactInequalityCooperation_2014,Fehr_TheoryFairnessCompetition_1999,Fehr_ReciprocalFairnessCooperation_2002}. Although the emergence and evolution of cooperation is a well-studied topic, how general cooperation evolves under a trade-off between efficiency and inequality is still an open question \citep{Axelrod_EvolutionCooperation_1981,Axelrod_FurtherEvolutionCooperation_1988,Bowles_CooperativeSpeciesHuman_2011,Stewart_Collapsecooperationevolving_2014,Press_IteratedPrisonerDilemma_2012}. Thus, we expect the study of cooperation under iniquity through well-understood tools such as Iterated Prisoner's Dilemma to yield new insights into the conditions underlying cooperation in society.

\paragraph{How does iniquity affect the formation of groups and thus cooperation between different groups?}
A related issue has to do with cooperation among different groups and collective actors. It is true that cooperative game theory offers many insights into the conditions under which different groups cooperate \citep{Ray_CoalitionFormation_2014}. However, in line with the research on fairness in experimental economics \citep{Fehr_TheoryFairnessCompetition_1999,Colasante_ImpactInequalityCooperation_2014}, one should expect the formation of groups and the stability of their cooperation to be sensitive to the trade-off between efficiency and inequality. Thus, we highlight coalition formation and cooperation among groups under iniquity as a fruitful avenue for future research.

% \section{Tolls and Inequality: Empirical Findings}
% \label{sec:data}
% \input{data}
%
% \section{Efficiency and Equality: Ramifications for Distributive Justice and Related Questions}
% \label{distributive-justice}
% \input{justice}

\section{Conclusion}
\label{sec:conclusion}
As cities are expected to house two thirds of the world population by 2030, the dual questions of the efficiency and the fairness of its most basic systems are essential. The field of algorithmic game theory has provided comprehensive results on the first of the two in the general model of congestion games but has so far been relatively silent on the second. There is no dispute that congestion is one of the most pressing issues that ``smart'' cities need to address and indeed, following the rich literature produced on tolling mechanisms and their contribution to increased efficiency, places such as Singapore or London are actively and successfully setting congestion control systems that are variants of these mechanisms.\footnote{Singapore will in fact implement from 2020 an individual continuous tracking system --- called ERP 2.0 --- for all the cars on its roads that will toll users depending on the duration and distance of their trip.}

Our paper contributes to the conversation by proving a general result on the topic of fairness, the Iniquity Theorem, stating that this increased efficiency comes at the cost of a greater inequality among the agents.
% Essentially, as tolls are introduced in the system, agents with higher income and a lower value of time versus money utilize the tolled links and ``price out'' the agents with lower incomes by the congestion they create.
% This central result is subsequently expanded in several directions. We define the iniquity as the marginal effect that the routing game has on the income distribution of the agents. We prove that for a natural choice of cost functions, this iniquity is scale invariant and stable if all agents follow no-regret algorithms. A system designer anxious to compromise between the efficiency and inequality of the game can make use of an algorithm presented to implement this trade-off. Finally, using a dataset collected during a large experiment in Singapore, we show the results of our Theorem at play in a real population and discuss its implications in the realm of distributive justice.
Following work may focus on two different threads.

On the one hand, it is possible that results akin to the Iniquity Theorem exist in different economic scenaria. The relation between inequality and efficiency has in fact been discussed extensively in economics, see \citep{baland1997wealth,bardhan2000wealth} for two informative if specialized works. Certain theoretical results (the ``neutrality theorem'') suggest that there is no connection; others that efficiency is favored by inequality (but mostly because the existence of strong firms and small decision groups favors it, not inequality between individuals).  As for the roots of inequality, they have been sought in globalization, intensifying gender and race gap, the decline of unions and the rise of neoliberalism, access to education, the creation of peculiar wealth in the financial industry, ``rent seeking'' (that is to say, influence peddling by the powerful), even in the explosive growth of the IT industry.

On the other hand, our community may focus on the redistributive aspects of its mechanisms. Numerous cities have the stated intention of channeling the money collected via tolls into the operation of its public transport systems, such as the Move NY plan debated in New York City. This thread may inch closer to the question of network design and could effectively reverse the consequence of inequality after tolls are levied. Future work could include designing objective functions similar to that used in our trade-off algorithm that include the utility of improved infrastructure --- namely, public transportation --- for lower income citizens and provide intuition on optimal tolling mechanisms to balance efficiency and fairness.

% \section{Extensions (not to be compiled in final paper!)}
% \input{extensions}

\subsubsection*{Acknowledgements}
Kurtulu\c{s} Gemici acknowledges NUS Strategic Research Grant (WBS: R-109-000-183-646) awarded to Global Production Networks Centre (GPN@NUS).
Elias Koutsoupias acknowledges
 ERC Advanced Grant 321171 (ALGAME).
Barnab\'e Monnot
acknowledges the SUTD Presidential Graduate Fellowship.
 Christos Papadimitriou acknowledges NSF grant 1408635 ``Algorithmic Explorations of Networks, Markets, Evolution, and the Brain".
Georgios Piliouras
acknowledges
SUTD grant SRG ESD 2015 097, MOE AcRF Tier 2 Grant 2016-T2-1-170, NRF grant NRF2016NCR-NCR002-028 and a NRF fellowship.

Barnab\'e Monnot and Georgios Piliouras would like to thank the other members of the National Science Experiment team at SUTD: Garvit Bansal, Francisco Benita, Sarah Nadiawati, Hugh Tay Keng Liang, Nils Ole Tippenhauer, Bige Tunçer, Darshan Virupashka, Erik Wilhelm and Yuren Zhou. The National Science Experiment is supported by the Singapore National Research Foundation (NRF), Grant RGNRF1402.

\bibliographystyle{apalike}
\bibliography{refer}

\newpage
\appendix
\section{Appendix}
\label{sec:appendix}
\subsection{Iniquity in the asymmetric case}
\label{sec:asymmetric}
Does the Iniquity Theorem hold for the asymmetric nonatomic routing games? In general, the answer is no, with a few caveats. First, asymmetric routing games, by their very nature, introduce populations that are not comparable: they may have different source and destinations as well as different available paths to reach one from the other. As shown in the remainder of this Appendix, it is then not possible to find a strict equivalent to the Iniquity Theorem in the asymmetric case. But we are perhaps asking too much: all counterexamples directly relate to the fact that the subpopulations in different commodities are not always comparable.

There is a way out: \textit{among each subpopulation, the inequality worsens.} This result is a straightforward application of the Iniquity Theorem in the symmetric case for each commodity. As tolls are introduced, users in each commodity polarize such that higher incomes use tolled roads, following the construction presented in Section \ref{sec:inequality}.

This Appendix shows how two constructions of measuring inequality fail to give a strict equivalent to the Iniquity Theorem in the asymmetric case. Along the way, we provide some intuition hinting to why this should be the case.

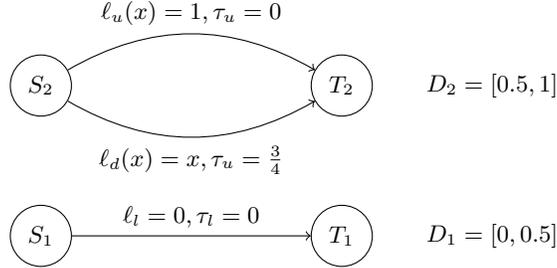
\begin{figure}
  \centering
    \begin{tikzpicture}
      \node[state] (1) at (0, 0) {\(S_2\)};
      \node[state] (2) at (4, 0) {\(T_2\)};
      \node[state] (3) at (0, -2) {\(S_1\)};
      \node[state] (4) at (4, -2) {\(T_1\)};
      \node at (6, 0) {\( D_2 = [0.5, 1] \)};
      \node at (6, -2) {\( D_1 = [0, 0.5] \)};
      \draw[->] (1) edge [bend left] node [above] {$\latency{u}(x)=1, \toll_u = 0 $} (2);
      \draw[->] (1) edge [bend right] node [below] {$\latency{d}(x)=x, \toll_u = \frac{3}{4} $} (2);
      \draw[->] (3) edge node [above] {\( \latency{l} = 0, \toll_l = 0 \)}(4);
    \end{tikzpicture}
  \caption{A multicommodity example: the interval \( D_1 \) of agents with income between 0 and 0.5 wishes to go from \( S_1 \) to \( T_1 \), while the interval \( D_2 \) goes from \( S_2 \) to \( T_2 \). Tolls and latencies are given along the edges.}
  \label{fig:pigouasym}
\end{figure}

\paragraph{First approach}
In this paragraph, the same approach as the symmetric case is employed: agents follow an \textit{ex ante} income distribution \( \ii \) and incur a cost due to playing the asymmetric game \( \Gamma \), recorded in the \textit{ex post} distribution \( \epii = \ii - \alpha \cdot \cost^F \).

An asymmetric nonatomic routing game is presented in Figure \ref{fig:pigouasym}. Agents follow income distribution \( \ii(x) = x \). This game has the property that the Gini improves \textit{after} the game is played (provided \( \alpha \) is small enough that resulting incomes remain nonnegative). Intuitively, it is always possible to construct an example where a mass of players with lower income has zero cost in the game while players with higher income have a positive cost, thus artificially bringing down the value of the Gini coefficient by reducing the gap between lower and higher incomes.

\paragraph{Second approach}
One possible alternative to find a meaningful interpretation of the Iniquity Theorem in the asymmetric case is to study the following:
\begin{itemize}
  \item First we look at the Nash Equilibrium of the game with no income effect, i.e. the edge cost functions have the familiar shape of \( f_e(z, \toll_e) = \latency{e}(z) + \toll_e \), for some latency functions \( \latency{e} \). Call \( \ii_0 \) the resulting income distribution after the game is played and agents' income is \( \ii_0(x) = \ii(x) - \alpha \cdot \cost^F(x) \).
  \item Second, the game is played with income effects as defined by the edge cost functions follwoing assumptions~\eqref{eq:edgecf}. Call now \( \epii \) the resulting income distribution.
\end{itemize}

This would fix the example proposed in Figure \ref{fig:pigouasym}. Here, \( \ii_0 \) is such that half of the agents are playing the lower game and incurring a cost of 0, while the other half is split equally between players on the bottom link and the top link incurring a cost of 1. In \( \epii \), the bottom half of players still incurs a cost of 0. However, the distribution changes for incomes 0.5 and above, as only \( 1-\frac{\sqrt{3}}{2} < \frac{1}{4} \) of the richest agents will use the link with a toll and incur a cost strictly smaller than 1. We find here a situation similar to the Iniquity Theorem in the symmetric case and can show that the Gini coefficient is indeed worse from \( \ii_0 \) to \( \epii \).

However this example only works because the flow of one population does not affect the other population, and so in a sense is not general enough. Before showing a counterexample where two populations interact and the Gini coefficient is improved from \( \ii_0 \) to \( \epii \), we give some intuition explaining how this second approach relates to the symmetric case studied in Section \ref{sec:inequality}. More precisely, we show that looking at the gap between \( \ii_0 \) and \( \epii \) in the symmetric case is strictly equivalent to considering \( \ii \) and \( \epii \). This gives us a sense that this alternative approach is satisfying because it does reduce to the Iniquity Theorem in the symmetric case.

\paragraph{Relation to the symmetric case:} The asymmetric case is such that we have \( K \) commodities routed through the network between \( K \) origin-destination pairs \( (s_k, t_k)_k \), instead of 1 commodity. We want to show that if the Iniquity Theorem holds for the asymmetric case as described above, it would automatically hold for the symmetric case. We give in Table \ref{tab:asymdiff} the differences between the two cases.

\begin{table}
  \centering
  \begin{tabular}{|p{2cm}|p{5cm}|p{5cm}|}
    \hline
      & Ex ante distribution & Ex post distribution\\
    \hline
    Symmetric case & Original income distribution \( \ii \) & \( \epii = q - \alpha \cdot \cost^F \) for game with income \\
    \hline
    Asymmetric case & Income distribution in NE of game with no income \( \ii_0 \) & \( \epii = q - \alpha \cdot \cost^F \) for game with income \\
    \hline
  \end{tabular}
  \caption{Differences between the symmetric and asymmetric case}
  \label{tab:asymdiff}
\end{table}

The Iniquity Theorem in the symmetric case states that moving from \( \ii \) to \( \epii \) is worse from the point of view of the Gini coefficient. We want to show that moving from \( \ii \) to \( \epii \) is worse from the point of view of the Gini coefficient if and only if moving from \( \ii_0 \) to \( \epii \) is worse from the point of view of the Gini coefficient. In other words, for the symmetric case, the result holds whether we use \( \ii \) or \( \ii_0 \) as our \textit{ex ante} distribution.

The proof is the following: in the Nash equilibrium of the game with no income and only one commodity, it is easy to see that all players will incur the exact same cost \( C \). Therefore, the original income distribution \( \ii \) and the income distribution in the NE of the game with no income \( \ii_0 \) are equal up to a translation of \( C \): \( \ii_0 = q - C \). By rescaling this distribution appropriately so that it has equal mean to the distribution \( \epii \), we can compare them and see that \( G(\epii) \geq G(\ii_0) \). This result is plotted in Figure \ref{fig:symmcase}.

\begin{figure}
  \centering
  \begin{tikzpicture}
    \definecolor{myred}{rgb} {0.941,0.3294,0.1921}
    \definecolor{myblue}{rgb} {0.2705,0.6627,0.8705}
    \definecolor{mygreen}{rgb} {0.7490,0.8078,0.5019}
    \tikzstyle{nome}=[anchor=west, minimum height=\altura,minimum width=2cm,text width=1.8cm]
    \draw[thick,->] (0,0) -- (0,4);
    \draw[thick,->] (0,0) -- (4,0);
    \draw[thick,myblue] (0,1) -- (4,3);
    \draw[thick,myred] (0,0.75) .. controls (0,2.75) and (4,0.75) .. (4,2.75);
    \draw[thick,mygreen] (0,0.5) -- (4,2.5);
    \draw[thick,->] (1,1.4) -- (1,1.1);
    \draw[thick,->] (3,2.4) -- (3,2.1);
    \node[myblue,nome] at (4.2, 3.2) {\textit{Ex ante} \( \ii \)};
    \node[myred,nome] at (4.2, 2.75) {\textit{Ex post} \( \epii \)};
    \node[mygreen,nome] at (4.2, 2.3) {\textit{Ex ante} \( \ii_0 \)};
  \end{tikzpicture}
  \caption{Using \( \ii \) or \( \ii_0 \) as \textit{ex ante} distribution in the Iniquity Theorem for the symmetric case does not change the result, since the shape of the quantile function matters for comparing two distributions with equal means.}
  \label{fig:symmcase}
\end{figure}
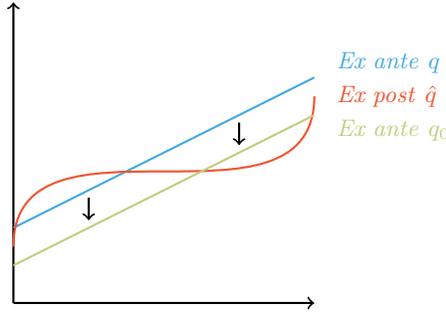

\paragraph{Introduction of \( \text{ }\Gamma_2 \)}
Considering the gap between \( \ii_0 \) and \( \epii \) is thus equivalent in the symmetric case. We now show that there are natural counterexamples that lead to an improved Gini coefficient for \( \epii \) and so the Iniquity Theorem does not hold again for the asymmetric case in this alternative.

Introduce game \( \Gamma_2 \) of players in the interval \( [0, 1] \) with income \( \ii(x) = x \). There is a value \( x^* \in [0, 1] \), such that all players \( x \leq x^* \) make up the first commodity going from \( s_1 \) to \( t_1 \) and all players \( x \geq x^* \) make up the second commodity going from \( s_2 \) to \( t_2 \). The network is given in Figure \ref{fig:pigouasym2}.

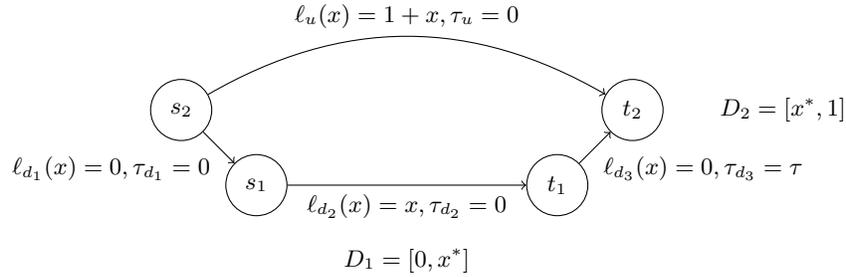
\begin{figure}
  \centering
    \begin{tikzpicture}
      \node[state] (1) at (0, 0) {\(s_2\)};
      \node[state] (2) at (6, 0) {\(t_2\)};
      \node[state] (3) at (1, -1) {\(s_1\)};
      \node[state] (4) at (5, -1) {\(t_1\)};
      \node at (8, 0) {\( D_2 = [x^*, 1] \)};
      \node at (3, -2) {\( D_1 = [0, x^*] \)};
      \draw[->] (1) edge [bend left] node [above] {$\latency{u}(x)=1+x, \toll_u = 0 $} (2);
      \draw[->] (1) edge node [below left] {$\latency{d_1}(x)=0, \toll_{d_1} = 0 $} (3);
      \draw[->] (3) edge node [below] {$\latency{d_2}(x)=x, \toll_{d_2} = 0 $} (4);
      \draw[->] (4) edge node [below right] {$\latency{d_3}(x)=0, \toll_{d_3} = \toll $} (2);
    \end{tikzpicture}
  \caption{A multicommodity example: the interval \( D_1 \) of agents with income between 0 and \( x^* \) wishes to go from \( S_1 \) to \( T_1 \), while the interval \( D_2 \) goes from \( S_2 \) to \( T_2 \). Tolls and latencies are given along the edges.}
  \label{fig:pigouasym2}
\end{figure}

The main idea behind this example is to alleviate the load for the first type of agents in \( \epii \) by sending some agents from the second commodity to the upper edge, with the toll \( \tau \) as control variable. It has a very close relationship with \textit{subgroup decomposability} that we introduce in the next paragraph before expanding on the counterexample.

\paragraph{Subgroup decomposability} Most income inequality measures satisfy the four properties of scale and population independence, anonymity and the transfer principle. An additional property, \textit{subgroup decomposability}, is verified by some inequality measures, including the Atkinson index, but not by others, such as Gini.

The subgroup decomposability gives us an intuition as to why the Iniquity Theorem may fail in the asymmetric case. For the Atkinson index, it states that if our sample dataset is split in two, the index of the whole population may be computed as a weighted sum of the indices of the subset populations plus the index of the sub-population means.
\[
  A(y \cup z) = w_y A(y) + w_z A(z) + A(\mu_y, \mu_z)
\]

In the asymmetric case, our subpopulations are each group of agents \( D_k \) in one commodity \( k \). We know, by the symmetric case, that within each population, the inequality worsens, so \( A(\hat{D}_k) \geq A(D_k) \) if \( \hat{D}_k \) is the distribution of incomes in commodity group \( k \). However, this does not directly imply that \( A(\bigcup_{k \in K} \hat{D}_k) \geq A(\bigcup_{k \in K} D_k) \), since the additional \( A(\mu_1, \dots, \mu_K) \) term may have moved one way or the other. In other words, within each subgroup the inequality may have worsened, but in the overall population, it is possible that averages of each group become more equal and thus better the social inequality. Though the Gini coefficient does not satisfy subgroup decomposability, it is this line of thought that the counterexample mines.

\paragraph{Solution of \( \text{ }\Gamma_2 \)}
Set \( \alpha = 0.01 \). First, in the Nash Equilibrium of the game without tolls, all agents use links \( (d_1, d_2, d_3) \) and incur cost \( x \cdot 1 \), their income times the latency.

We now want to set \( \toll \) such that flow \( D_2 \) is further split in two: \( f \) units are sent to the upper path \( u \) while \( 1 - x^* - f \) are sent to the lower path \( d \). Call \( h^* \) the income quantile of the agent at \( x^* + f \). This income quantile is indifferent between going to the upper or the lower path, and thus is a solution to
\begin{align*}
  \cost_u(x) = \cost_d(x) & \Leftrightarrow & x(1+f) = x(x^*+1-x^*-f) + \toll \\
    & \Leftrightarrow & x + x f = x - x f + \toll \\
    & \Leftrightarrow & \toll = 2xf
\end{align*}

This yields a closed form for the toll \( \toll = 2 h^* (h^* - x^*) \) that depends on \( h^* \) and \( x^* \) alone. Note that we are not trying to obtain the optimal toll, but simply the toll \( \toll \) that divides population \( D_2 \) along a prescribed flow \( (f, 1-x^*-f) \).
We focus on the flow that is such that \( D_2 \) is split exactly in half, i.e. \( h^* = \frac{1+x^*}{2} \), and find that for a range of values \( x^* \), the Gini coefficient of \( \epii \) is indeed lower than that of \( \ii_0 \), which implies that the inequality is higher in \( \ii_0 \). We give in Figure \ref{fig:asym2gini} a graph showing the Gini coefficient for \( \ii_0 \) and \( \epii \) for different values of \( x^* \).

However, as prescribed by the Iniquity Theorem in the symmetric case, restricting our view to population \( D_2 \) alone, the Gini coefficient has indeed worsened. We present the difference \( G(\epii) - G(\ii_0) \) in Figure \ref{fig:asym2diff}.

\begin{figure}[!ht]
	\centering
	\begin{minipage}[t]{0.42\linewidth}
		\centering
		\includegraphics[width=0.95\textwidth]{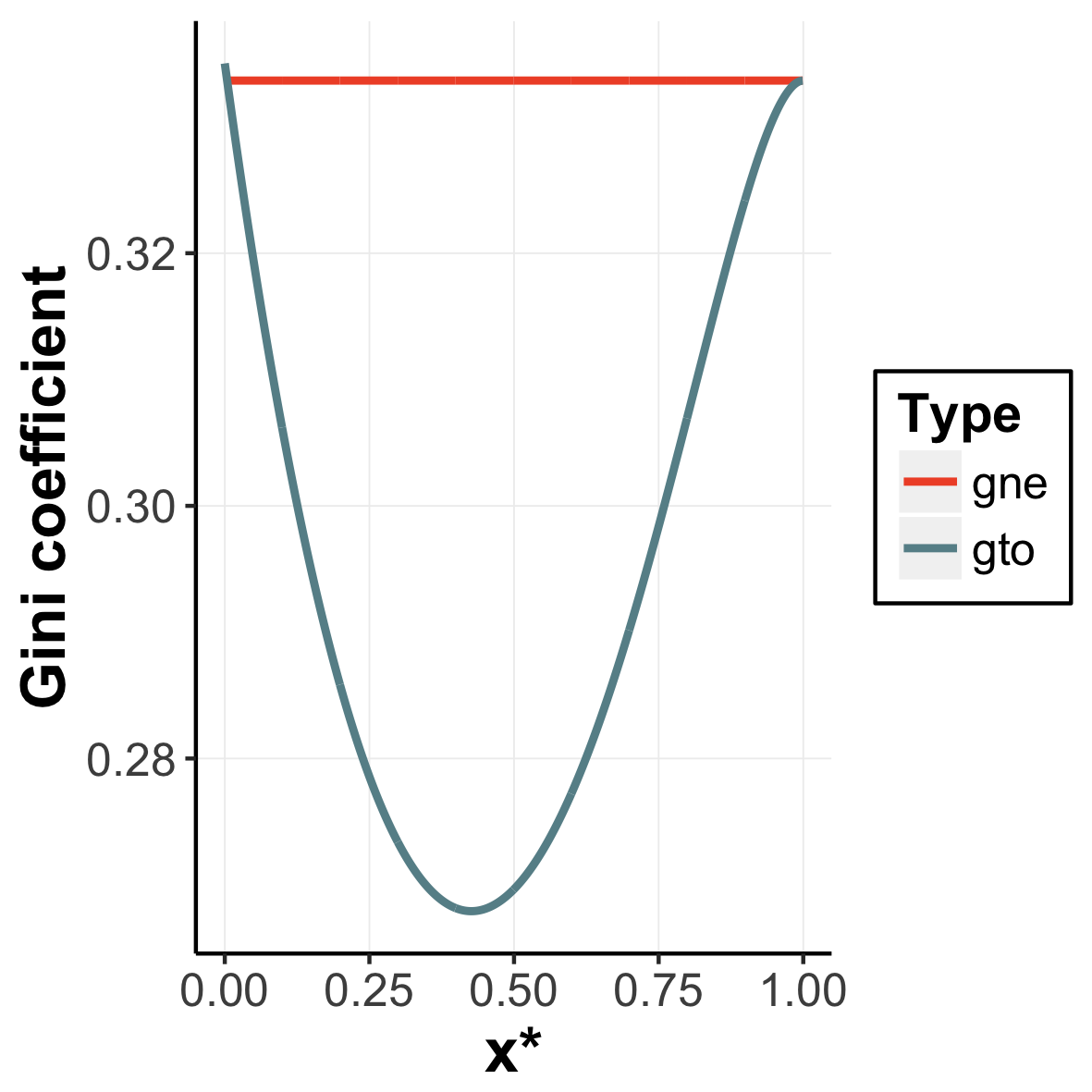}
    \caption{Line plot of \( G(\epii) \) and \( G(\ii_0) \) along \( x^* \). For most values of \( x^* \) the Gini improves from \( \ii_0 \) to \( \epii \).}
    \label{fig:asym2gini}
	\end{minipage}
  \begin{minipage}[t]{0.1\linewidth}
    \begin{tikzpicture}
    \end{tikzpicture}
  \end{minipage}
	\begin{minipage}[t]{0.42\linewidth}
		\centering
		\includegraphics[width=0.95\textwidth]{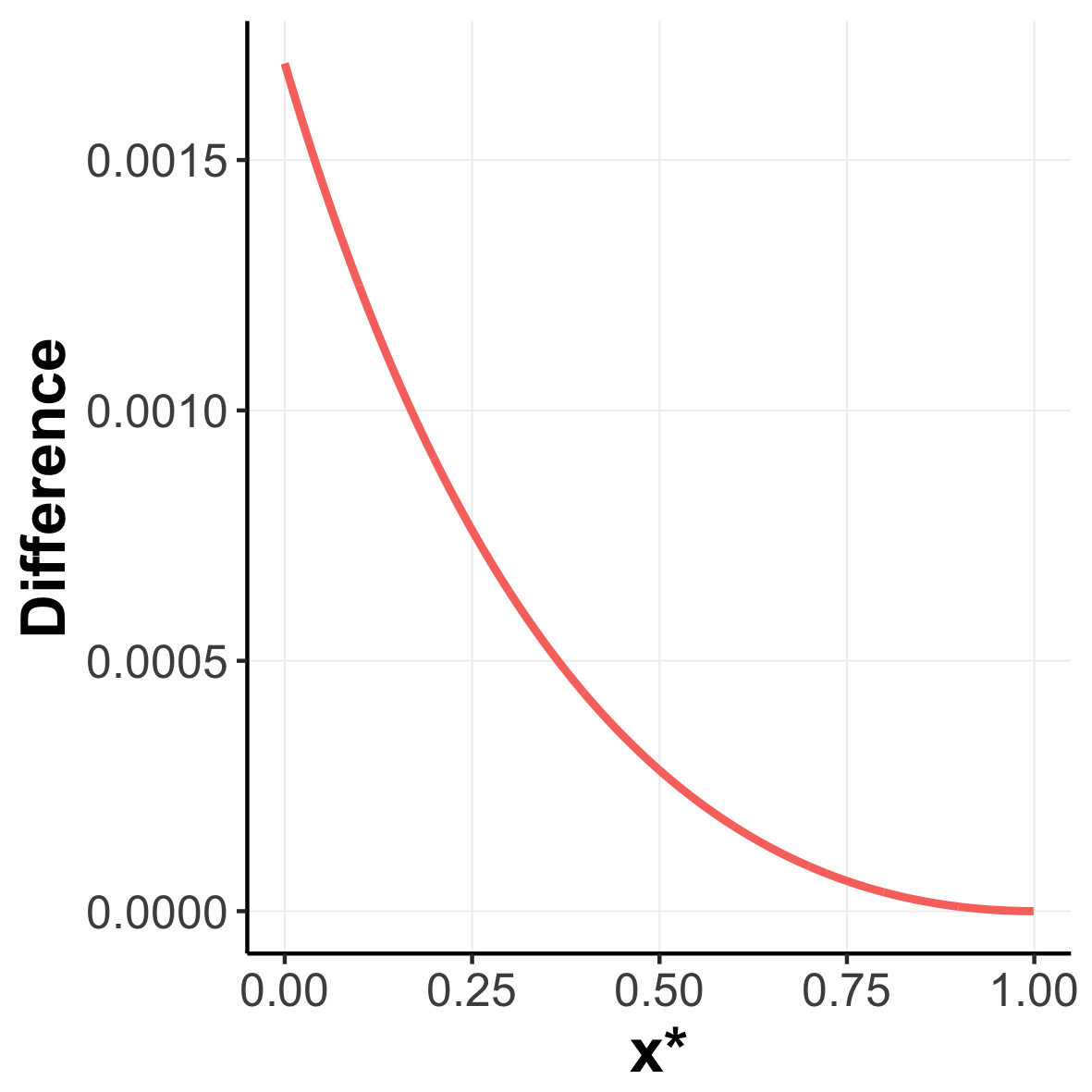}
    \caption{Difference between \( G(\epii) \) and \( G(\ii_0) \) along \( x^* \) restricted to population \( D_2 \). The difference is always positive which indicates that the inequality within population \( D_2 \) has increased.}
    \label{fig:asym2diff}
	\end{minipage}
\end{figure}

\subsection{Computing the efficiency-equality trade-off}
\label{sec:tradeoff}
\paragraph{The model}
Because of the computational nature of this section, and for the sake of simplicity,
we will stick to a simplified, discrete model.  Very few of these simplifications are crucial.
We assume a population whose income is presented in $n$ {\em quantiles} $\ii_1,\ldots,\ii_n$,
where $\ii_1$ stands for the average income of the lowest $1\over n$ of the population
--- if $n$ were $100$, these would be the income percentiles.

We have $K$ parallel links ---
we assume that $K$ is fixed.  Each link $e$ has a delay function $f_e(x)$ which we assume
for simplicity to be piecewise constant with increments at values of $x$ that are multiples of
$1\over n$ (so that each link accommodates full quantiles), and that the delays have integer
values in the set $[D]$, where $D$ is the maximum delay.
Evidently, the problem is one of allocating each quantile to a link, and imposing
appropriate tolls.  It is easy to see that at equilibrium each link will be assigned a
{\em contiguous} set of quantiles.

\paragraph{The objective}
We seek to optimize a trade-off between efficiency and equality,
that is to say, a weighted sum of total delay and the Gini coefficient resulting from
this game, say of the form ``minimize total delay + $\lambda$ times the Gini after the game,''
where $\lambda>0$ is the relative importance of equality over efficiency.

It is important to note that the Gini coefficient before the game is in this case captured by (ignoring
additive terms and a factor of $-2\over n$)
$$\sum_{i=1}^n (n+1-i)\ii_i\over \sum_{i=1}^n \ii_i.$$
This is on account of the fact that, in the sum that approximates the double integral in Equation \eqref{eq:lorenz},
the lowest quantile appears $n$ times, the second lowest $n-1$, etc.

After the game imputes a cost to the $i$th quantile, the Gini coefficient is captured by
$$\cdot{\sum_{i=1}^n (n+1-i)(\ii_i-\delta_i)\over \sum_{i=1}^n (\ii_i-\delta_i)},$$
where $\delta_i= \ii_i d_i + \toll_i$ is the cost of the equilibrium to the $i$th percentile,
and $d_i$ is the delay and $\toll_i$ is the toll incurred by the $i$th quantile.   Now, since
it is reasonable to assume that $\delta_i<<\ii_i$, this quantity can be adequately
represented by its numerator divided by the sum of the $\ii_i$'s\footnote{For more accuracy, the
computed value of $\sum_i \delta_i$, can be plugged in here and repeat the computation.}.
Thus, omitting constant terms (it is important to recall that the
$q_i$'s {\em are} constant), both additive and multiplicative, we conclude that what is minimized is
a linear function of the delays $d_i$ and the tolls $\toll_i$.  Adding to them the total
delay\footnote{Note that even the {\em total weighted delay} $\sum_i \ii_i d_i$ can be similarly
accommodated as part of the trade-off.}, we conclude that the objective is of the form
$$\min_{{\rm allocation\ of\ quantiles\ to\ links}} \sum_{i=1}^n (\alpha_i d_i + \beta_i \toll_i),$$
for some known positive parameters $\alpha_i$, $\beta_i$.

\paragraph{The algorithm}
The algorithm is dynamic programming; namely, we compute the quantity cost$[S,m,d]$
with $S\subseteq [K]$, $m\leq n$, and $d\leq D$, which is the smallest value of the
objective that can be achieved by allocating the {\em lowest} $m$ percentiles to the set
$S$ of links (in the optimum order) with the (largest) delay of the $m$-th percentile equal to $d$. The algorithm is presented in Algorithm \ref{alg:tradeoff}.

\begin{algorithm}
    \SetAlgoNoLine
    \DontPrintSemicolon
    \KwData{Calculate the values cost$[\{e\},m,d]$ for all links $e, m\in [n], d\in [D]$}
    \Begin{\For{\( s \gets 2 \) \KwTo \( K \)}{
        \For{All sets $S\subseteq L$ with $|S|=K$}{
            \For{\( m \gets 1 \) \KwTo \( n \)}{
                \For{\( d \gets 1 \) \KwTo \( D \)}{
                    cost$[S, m, d] = \min_{e \in S,\, r < m:\, \ell_e(r) = d;\, d' \leq d}$ cost$[S-\{e\}, m-r, d']$

                    \quad\quad\quad $+ \sum_{j=m-r+1}^m (\alpha_j d' + \beta_j t(d,d',r,m-r))$
                }
            }
        }
    }}
\caption{A dynamic programming algorithm to compute the trade-off between efficiency and equality.}
\label{alg:tradeoff}
\end{algorithm}

By $\toll(d,d',r,m-r)$ we denote the toll required to equalize, for the $m-r$th quantile, the delay $d'$ with the greater delay $d$.  In conclusion (here $D^*\leq n$ is the number of different values of the delay in the network):

\begin{theorem}
    \label{thm:tradeoff}
The optimum trade-off between total delay and the Gini coefficient can be computed in time $O(nD^*)$
\end{theorem}

But of course, the $O$-notation hides the constant $K^2 2^K$.

\subsection{Iniquity in the Pigou example for cost~\eqref{eq:cf1}}
\label{sec:pigoucf1}

If instead we focus on the perceived latency $\cost(x)/2x$ (actual
latency plus tolls over income) as per Equation~\eqref{eq:cf1}, we get the associated social
perceived latency
$\int_0^1 \cost(x)/2x\, dx=1 - \s+ \s^2 - \s^2 \ln \s$, shown in
Figure~\ref{fig:Pigou-social-latency}. The perceived latency is
similar to social cost, but distorted because we divide the social
cost of each participant by their income.

\begin{figure}
  \centering
  \begin{minipage}[c]{0.46\textwidth}
		\centering
		\includegraphics[width=.95\textwidth]{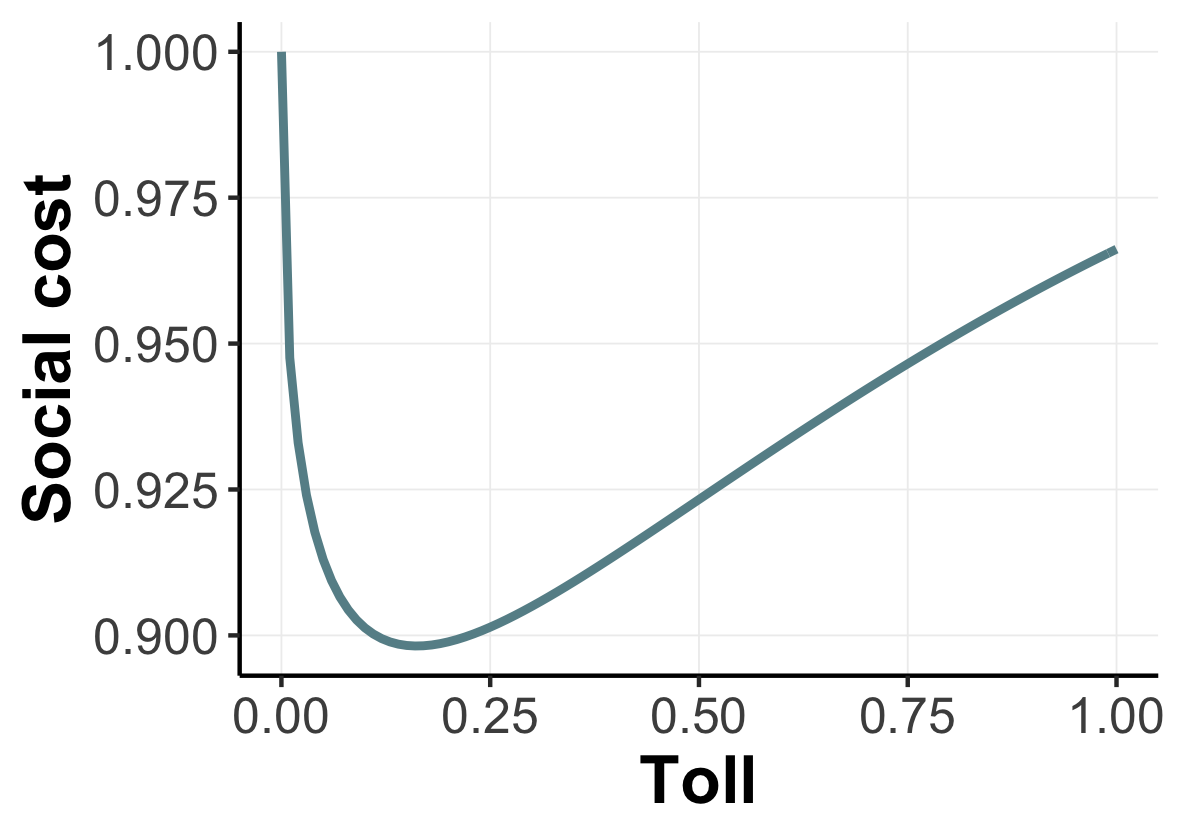}
	\end{minipage}
	\begin{minipage}[c]{0.46\textwidth}
    \centering
    \includegraphics[width=.95\textwidth]{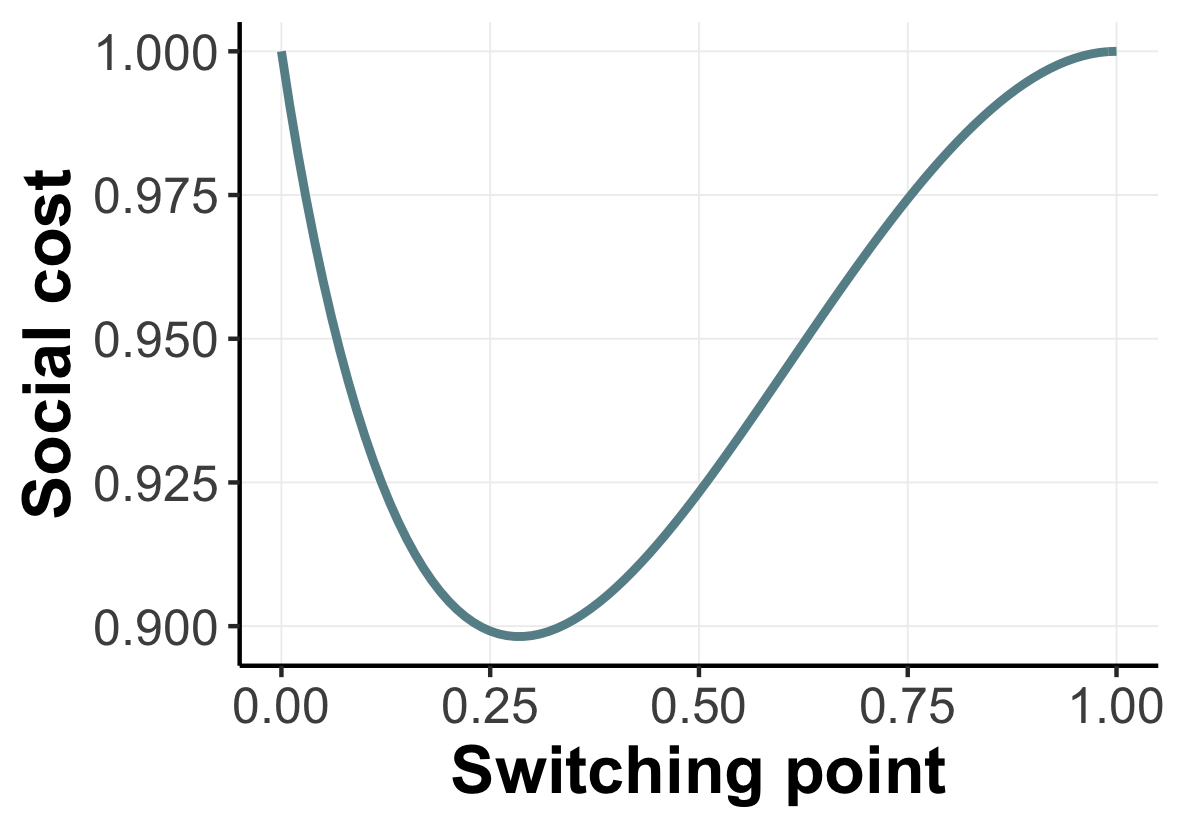}
  \end{minipage}
  \caption{The perceived latency (time wasted in transportation plus
    tolls over income) of the Pigou network with income
    distribution $\ii(x)=2x$, as a function of the tolls (left) and switching
    point $\s$ (right).
    % The minimum value $23/27$ achieved at $\toll=2/9$ and $\s=1/3$.
    % Note that the minimum value is less than $1$, the best social cost
    % with tolls when we don't take into account the income distributions
    % (or equivalently in a society with Gini coefficient 0).
  }
  \label{fig:Pigou-social-latency}
\end{figure}

\subsection{Tolls and Inequality: Empirical Findings}
\label{sec:data-app}

We tap into a remarkable dataset from Singapore, the National Science Experiment (NSE), to examine whether the theoretical predictions of the Iniquity Theorem are empirically supported. Singapore, a postcolonial country with a rapid trajectory of economic development since independence \citep{huff_developmental_1995, trocki_singapore_2006}, provides an ideal setting to analyze the effects of tolls. The city-state is one of the most affluent countries in the world; and, it is also the third most densely-populated country after Macao SAR and Monaco.%
\footnote{The World Bank, World Development Indicators, 2018.} %
However, despite being a densely populated country, Singapore avoids congestion on surface streets and highways. Several policies achieve such an outcome. First, the city-state has a modern and well-designed road network. Second, the Singaporean state discourages car ownership by making it fairly expensive through various measures, as hinted in the introduction. Third, most importantly, Singapore's Land Transport Authority (LTA) actively aims to minimize road congestion through its Electronic Road Pricing System (ERP). Private taxicabs and cars using the main arteries of the road network are charged tolls that can add up to significant amounts for a single journey, particularly during peak morning hours.
The conditions obtained in Singapore's transportation landscape can be considered a natural experiment \citep{Dunning_NaturalExperimentsSocial_2012} as the absence of congestion offers a stark contrast to other densely populated metropolitan centers and it is possible only because of preventive measures such as significant road tolls. In what follows, we offer evidence on the Iniquity Theorem by studying the distribution of delays in a representative sample of students.

%Singapore is often considered a social experiment because of the outsized role of its infrasructurally powerful state organization in every aspect of social life.

Singapore is an ideal setting to study the effect of tolls by using spatial data because the Singaporean state implements wide-ranging housing, urban development, and educational policies to prevent spatial segregation based on income inequality
\citep{Huat_PublicHousingResidents_2000, Chua_CommunitarianIdeologyDemocracy_1995, Phang_Affordablehomeownershippolicy_2010}. Sociological studies demonstrate that income inequality often leads to income segregation,%
\footnote{Income segregation is defined as ``the uneven geographic distribution of income groups within a certain area'' \citep[1093]{Reardon_IncomeInequalityIncome_2011}}
which in turn is associated with infrastructural and urban decay \citep{Sampson_MovingInequalityNeighborhood_2008, Jencks_SocialConsequencesGrowing_1990, Jargowsky_TakeMoneyRun_1996}. That is important because high levels of spatial segregation would complicate the causal relationship between tolls and delays. In Singapore, high-quality public housing projects across the city-state and evenly distributed infrastructure minimize the effects of such confounding factors.
%Thus, the average distance, road conditions, and congestion levels do not differ greatly among many of the wealth strata we study.

The dataset we use comes from the National Science Experiment (NSE) \citep{wilhelm2016sensg,Monnot2016}, which is a large experiment deployed by SUTD in collaboration with several Singapore public agencies to understand the use of the city by students aged 12 to 20. Sensors were made available to be carried by the students over a period of 4 days. Every 13 seconds, whenever possible, the sensor logs its geographical location computed from scanning surrounding WiFi access points as well as several environmental factors such as relative temperature and humidity or noise level. More information on the NSE dataset can be found in Appendix \ref{sec:nse-app}.

The NSE dataset is a valuable source of information to study important questions pertaining to everyday life, and in particular transportation and congestion.  Here we combined the NSE dataset with a dataset of property prices%
\footnote{Obtained from the Singapore Real Estate Exchange (SRX) online platform.} %
to infer the connection between transportation and equality. As in previous studies with the NSE dataset \citep{monnot2017routing}, we restrict our attention to trips taken to go to school or university in the morning. The motivations are four-fold: first, students tend to take the most direct route to school in the morning, while the return trip will typically feature more stops in extraneous locations; second, within one category of schools, most students will start classes around the same time, making their trips comparable; third, students meaningfully interact with traffic as morning conditions present higher congestion during peak hours; fourth, almost all the major arteries and gateways incur ERP tolls in the morning hours.

A trip is split into segments, each segment corresponding to one distinct mode of transportation. The mode is recognized algorithmically between four options: walking, in a train, in a bus or in a car. It was shown that the mode recognition algorithm has a high level of accuracy \citep{Wilhelm2017}. Additionally, descriptive statistics of the dataset closely resemble survey data collected by the LTA: for instance, the share of public transportation users is around 60\% both in the survey and in our dataset.

The property price dataset covers a geographically large part of the city-state. We attribute each home a ``wealth index'' that reflects the average price per square foot of the surrounding area. We further divide this wealth index into 9 brackets. This is a reasonable strategy in Singapore, where most household wealth is invested in household's primary residence \citep{Edelstein_HousePricesWealth_2004,Abeysinghe_AggregateConsumptionPuzzle_2004}.%
\footnote{As \citet[571]{Abeysinghe_AggregateConsumptionPuzzle_2004} point out, ``even the most affordable public apartments in Singapore could cost 5–10 times the average annual household income.''} %
Although Singapore's income inequality is fairly high,%
\footnote{The Gini coefficient for Singapore after accounting for various transfers and taxes is 0.402. For further details, please see Key Household Income Trends for 2016, accessible at
\url{https://www.singstat.gov.sg/docs/default-source/default-document-library/publications/publications_and_papers/household_income_and_expenditure/pp-s23.pdf}.} %
most Singaporeans live in public housing that are scattered around the island.%
\footnote{Public housing prices vary by neighborhood, but this variation is much lower than the private segment of the housing market, where demand is driven by affluent Singaporean families, foreign investors, and Singapore's sizable expatriate population
\citep{Huat_PublicHousingResidents_2000,Phang_Affordablehomeownershippolicy_2010}.} %
We do not capture families with different incomes scattered around the city but the significant housing price variation across the island implies that the families we identify through residential address are highly likely to have different socio-economic status.
%Singapore's historical and business districts are located in the center south of the island, with relatively high property prices around the area. Decentralization efforts have made locations in the West (Jurong) and in the East (Tampines) other hotspots of activities where property prices are typically higher --- though less so than in the center. The remaining prices around the island vary mostly according to how close they are to the hotspots. There are obvious limitations of using property prices as a proxy of income, especially with such a coarse dataset --- e.g., the current study does not reflect typically higher prices found around train stations.  Additionally, a large majority of Singapore's population lives in public housing for which the price may be dictated by policy, even though location is a significant factor in determining the rent.  %However, relevant results appear to confirm the existence of an ordinal measure to compare the resources available to the students.
Our sample is divided into primary schools, secondary schools or junior colleges, and polytechnics or universities in roughly similar proportions. This sampling pattern and Singapore's successful polycentric urban development model \citep{Han_PolycentricUrbanDevelopment_2005} ensure that we have an even geographic distribution of schools across the island. %Given such conditions, it is not surprising that the average traveled distance does not vary greatly among majority of the wealth brackets we study, as can be seen in Figures \ref{fig:avgdistpersgm} and \ref{fig:popwealth}.

\begin{figure}
	\centering
	\begin{minipage}[c]{0.35\linewidth}
		\centering
		\includegraphics[width=\textwidth]{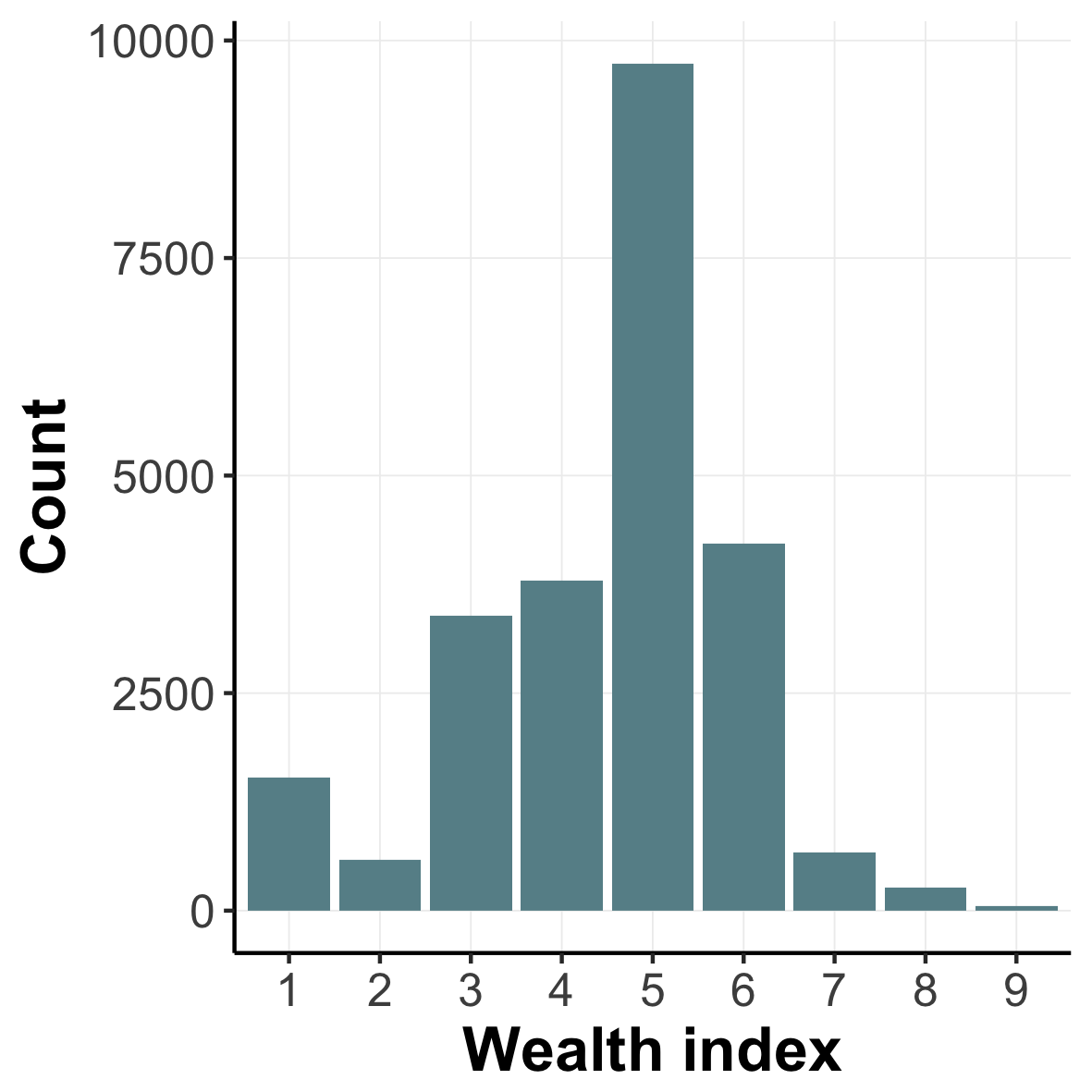}
		\caption{Population Per Wealth.}
		\label{fig:popwealth}
	\end{minipage}
	\begin{minipage}[c]{0.20\linewidth}
		\begin{tikzpicture}
		\end{tikzpicture}
	\end{minipage}
	\begin{minipage}[c]{0.35\linewidth}
		\centering
		\includegraphics[width=\textwidth]{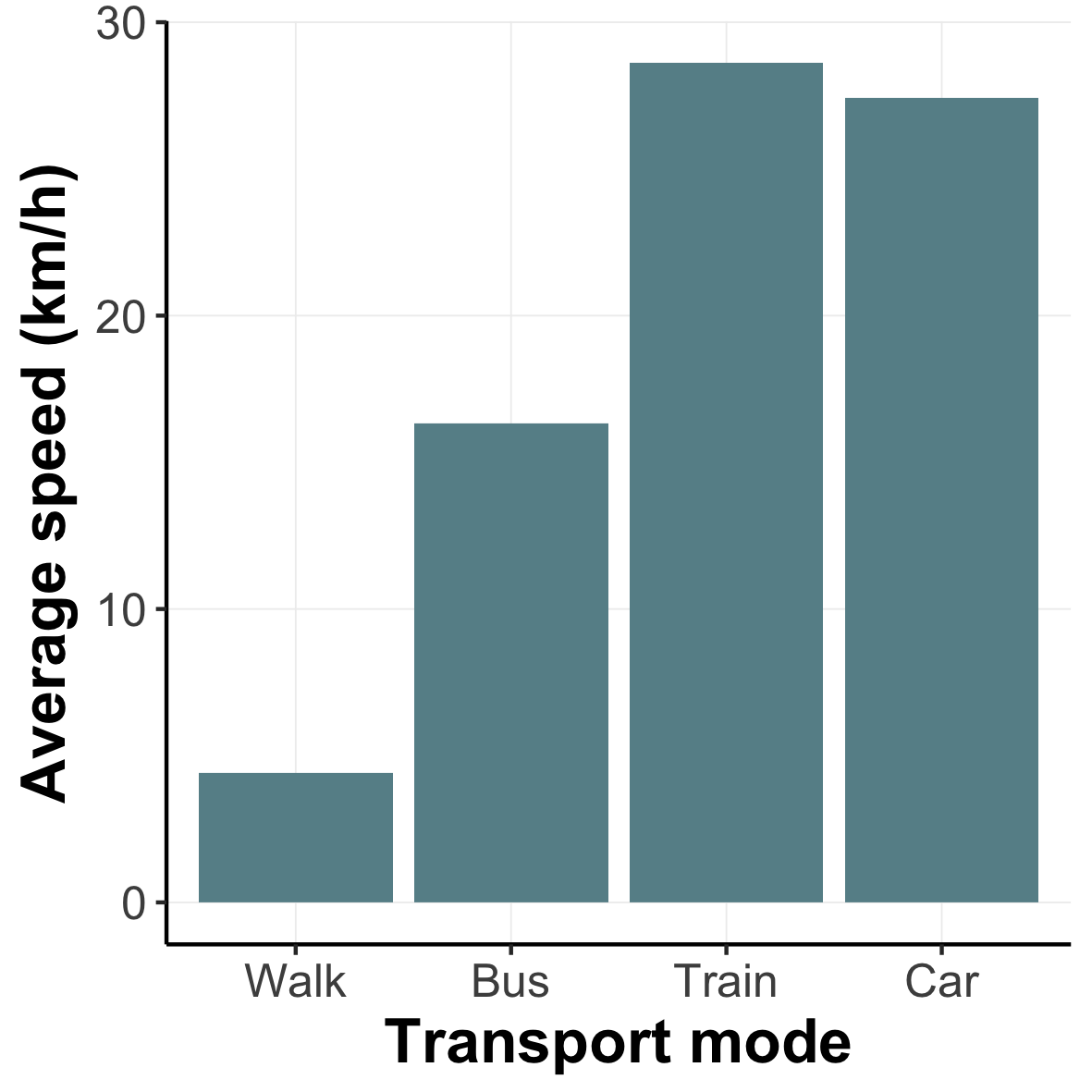}
		\caption{Average Speed Per Mode.}
		\label{fig:avgspeedpermode}
	\end{minipage}
\end{figure}

Based on residential location, the brackets 3, 4, 5, and 6 correspond to various middle-income neighborhoods, ranging from lower to upper-middle income groups. These brackets contain the great majority of our sample (Figure \ref{fig:popwealth}). Although there is some variation in the average distance traveled among these brackets, the differences are relatively small. The wealth brackets that markedly differ from the average of brackets 3--6 are the least and most affluent groups. This pattern follows from several factors. First, a large majority of Singapore's population lives in public housing, which are available to Singaporean residents even in the center and the belt surrounding the center of the city. The dense core and the surrounding periphery is where a majority of middle classes live, as can be seen in Figure \ref{fig:schperwealth} where red circles denote residences and yellow circles represent schools (circle size is proportional to student count). In these areas, families have access to a range of public and private schools, which has an equalizing effect on the average distance traveled. Second, the students from least affluent brackets are located predominantly in the north-west of the island, which is farther from the city core compared to the other residential areas we study (Figure \ref{fig:schperwealth}). Although this area contains its own share of educational institutions, students residing in the north-west of the island are likely to travel to other parts of the city to attend select schools.
%The supply of educational institutions in these areas do not match the rest of the city.
Third, it is likely that students from the most affluent brackets (located in the center-south of the city) attend private international schools or elite secondary schools in numbers disproportionate to the rest of the population. Some of these elite schools are located in campuses far away from the center of the city, which increases the average distance traveled for these wealth brackets.

%The segmented housing market, with a sharp price difference between public and private housing, affords us the opportunity to examine the differences between students coming from different socio-economic backgrounds. Private housing in Singapore is considerably more expensive than public housing and is accessible only to upper-middle class and upper-class Singaporean families.

% \begin{figure}[!ht]
% 	\centering
% 	\includegraphics[width=\textwidth]{../data/plots/schoolsPerBracket.png}
% 	\caption{Schools Per Wealth.}
% 	\label{fig:schperwealth}
% \end{figure}

% \begin{figure}[!ht]
% 	\centering
% 	\begin{minipage}[c]{0.46\linewidth}
% 		\centering
% 		\includegraphics[width=0.95\textwidth]{../data/plots/avgDistancePerSegment.png}
% 	\end{minipage}
% 	\begin{minipage}[c]{0.46\linewidth}
% 		\centering
% 		\includegraphics[width=0.95\textwidth]{../data/plots/avgDistancePerSegmentFacet.png}
% 	\end{minipage}
% 	\caption{The average trip distance per wealth bracket is presented in the two figures above. \textit{Left:} Average distance in each transport mode per wealth bracket. \textit{Right:} The left plot is split along the transport mode, to show the relative distances logged in each mode across wealth brackets.}
% 	\label{fig:avgdistpersgm}
% \end{figure}

The differences in the means of distance and duration among the three income groups are tested by using the non-parametric Kruskal-Wallis test as well as one-way ANOVA. These tests support (\/$p < 0.001$ for both tests) the hypothesis that there are significant differences across groups when it comes to travel distance and duration. We observe two patterns when we compare average distance traveled (Figure \ref{fig:avgdistpersgm}) with average duration (Figure \ref{fig:avgdurpersgm}) across different wealth brackets. First, there is an increase in car usage as we move from the least affluent to most affluent wealth brackets. That increase is accompanied by a secular decrease in walking. As a result, as we move from wealth bracket 1 to 9, we observe students cutting down their travel time because of the efficiency of cars as opposed to walking (Figures \ref{fig:avgdurpersgm} and \ref{fig:avgspeedpermode}).

% \begin{figure}[!ht]
% 	\centering
% 	\begin{minipage}[c]{0.46\linewidth}
% 		\centering
% 		\includegraphics[width=0.95\textwidth]{../data/plots/avgDurationPerSegment.png}
% 	\end{minipage}
% 	\begin{minipage}[c]{0.46\linewidth}
% 		\centering
% 		\includegraphics[width=0.95\textwidth]{../data/plots/avgDurationPerSegmentFacet.png}
% 	\end{minipage}
% 	\caption{The average trip duration per wealth bracket is presented in the two figures above. \textit{Left:} Average duration in each transport mode. It is notable that brackets 1 and 2 have respectively 7 and 5 minutes more travel time on average than the other brackets. \textit{Right:} The left plot is split along the transport mode, to show the relative durations spent in each mode across wealth brackets. Given that the least-affluent groups spend much more time on the roads compared to middle-income groups, there is a quasi-monotonic increase in the use of car as wealth increases, while we observe that the uses of walking and public transportation decrease as wealth increases.}
% 	\label{fig:avgdurpersgm}
% \end{figure}

Second, as can be seen in Figure \ref{fig:avgdurpersgm}, although the students from the two most affluent wealth brackets travel, on average, a longer distance in the morning hours compared to brackets 3--6 (Figure \ref{fig:avgdistpersgm}), the time they spend on the roads is only slightly more than the majority of the students coming from brackets 3--6. The higher average travel speed and thus lower travel duration enjoyed by the most affluent wealth brackets is even more pronounced when they are compared to the two least affluent wealth brackets (Figures \ref{fig:avgdistpersgm}, \ref{fig:avgdurpersgm}, and \ref{fig:avgspeedpermode}). Since the NSE dataset enables us to break down each student trip into different segments by the transportation mode, we can demonstrate that the advantage in travel duration enjoyed by the most affluent brackets is a result of the higher usage of private cars and taxicabs, as these transportation means reduce the use of less efficient means of transportation such as buses and walking. That is possible only because the road network in Singapore is efficient and not prone to congestion, a remarkable feat for such a densely populated country.

Although these findings support the Iniquity Theorem, they also provide some policy lessons that might be applicable in other contexts. It should be remembered that Singapore provides an efficient public transportation system that reduces the inequality created by the drive for efficiency. Furthermore, the Singaporean polycentric urban development model \citep{Han_PolycentricUrbanDevelopment_2005} also seems to work, since such a model addresses some of the root causes that lead to congestion in the first place.

\paragraph{Estimation of the iniquity in Singapore}
In addition to the findings above that show the relationship between income distribution and transportation delays in an empirical setting, an estimated travel cost for each trip in the dataset is obtained. We can now estimate how the distribution of transportation delays affects, in turn, the income distribution in the same empirical setting. For students using private transportation, we compute the cost estimate with the applicable taxi fare in Singapore.%
\footnote{LTA, \url{https://www.lta.gov.sg/content/ltaweb/en/public-transport/taxis\%20and\%20private\%20hire\%20cars/fares-and-payment-methods.html}.}
For students using public transportation, we compute the cost estimate with the applicable ticket fare, which is an increasing function of the distance.%
\footnote{PTC, \url{https://www.ptc.gov.sg/regulation/bus-rail/fare-structure}.}

We can associate the wealth index with a numerical value of the price per square foot for a private-market condominium apartment in the neighborhood. The same could be done using the price of a public housing apartment as there is a strong correlation between the two prices. The Iniquity Theorem would hold in either scenario.

The price per square foot is used as the income distribution \( \ii \) and the Gini coefficient is obtained for this distribution.
Equation~\eqref{eq:cf2} with the canonical agent cost~\eqref{eq:canonical} gives the resulting income distribution \( \epii \), for a value of \( \alpha \) small enough that the obtained coefficient is positive.
Letting \( \alpha \) go to zero yields the iniquity \( \frac{dG}{d \alpha} \). If this iniquity is positive, the distribution is more unequal after the game than before. Numerical experiments with the data show that the iniquity is indeed positive, at a value of 2.58, and indicate that the Iniquity Theorem holds for real data collected in Singapore.

Future work may focus on obtaining a more precise estimation of the iniquity. Indeed, the Gini coefficient of the current income distribution \( \ii \), at 0.09, does not reflect the one of Singapore, closer to 0.45 before taxes. This is due in part to the low granularity of using property prices as a proxy for income, effectively ``bundling'' households together under one value. Adding the tolls to the cost estimation of the trip would also put us closer to the true price of using private transportation in Singapore, but is not likely to change the qualitative result of a positive iniquity.

%In Figure \ref{fig:carvspub} we show the relative fraction of students using public transportation versus private transportation, broken by wealth index. We note that in the neighborhoods with higher housing prices, students are in general more likely to use private transport. The Spearman rank correlation value is 0.8 when comparing the series of percentages of private transportation users with a monotonic increasing sequence (P = 0.013).
%In Figure \ref{fig:avgspeed} we show the average speed of students in each wealth index group is obtained. Interesting, this plot is bimodal, with students in districts with much lower and much higher property prices maintaining a relatively higher average speed than students in average property prices districts. This finding warrants further examination.

\subsection{The National Science Experiment (NSE)}
\label{sec:nse-app}
The National Science Experiment is an initiative from the National Research Foundation (NRF), the Singapore University of Technology Design (SUTD) and partners to provide low-cost sensors to students in primary school, secondary school and junior colleges for up to four days. The sensors, called SENSg, log information at most every 13 seconds containing their GPS location --- inferred from scanning surrounding Wi-Fi access points --- as well as environmental factors such as relative temperature and humidity or noise levels. The sensors were given out in batches over a year, with a few thousand students participating in each batch.

As such, the NSE dataset is a rich description of individual user behaviour in Singapore, capturing the daily activity of about 50,000 students. Although the sample focuses on one particular class of the population in Singapore, previous studies \citep{monnot2017routing} were able to infer results on the congestion level as a whole.
% It is also useful to compare this experiment with other well-known field experiments such as Milgram's small world experiment \citep{travers1967small} that were able to extract meaningful trends from a highly biased sample population.
Large scale data sources have been employed for previous studies in Singapore, however all with a limitation of some sort. For instance, public transport data was used in \citep{Sun2012, Lee2014, Holleczek2015, Poonawala2016} and does not capture the behaviour of citizens in private transport vehicles, or GSM cellular phone data in \citep{Holleczek2014} that is not granular enough to accurately detect the sequence of transport modes during a trip.

In the current study, our dataset is composed of 24,104 individual morning trips that were made throughout 2017 by 12,823 distinct students. The trips are a subset of a larger dataset, filtered using the following conditions:
\begin{itemize}
  \item The trip should be between 5 and 60 minutes to account for too short or too long trips that may be logged inaccurately by sensor noise, which accounts for about 8\% of the original dataset.
  \item The trip should have a minimum of 20 sensor logs throughout, to ensure that enough granularity is provided and transport modes are recognized appropriately.
  \item The trip should not be included in that of a school-chartered bus, which can be recognized algorithmically by identifying identical trips between groups of students.
\end{itemize}

\subsection{Efficiency and Equality: Ramifications for Distributive Justice and Related Questions}
\label{sec:justice-app}

The Iniquity Theorem underlines a trade-off between efficiency
and inequality. The normative implications of the {\em Iniquity
Theorem} are manifold. However, the ramifications of the {\em Iniquity
Theorem} are not limited to normative issues. Inequality has serious
social and thus serious economic consequences. Decades of research show
that the ``social consequences of economic inequality are sometimes
negative, sometimes neutral, but seldom \ldots{} positive''
\citep[64]{jencks_does_2002}. In what follows, we outline certain
ramifications of the Iniquity Theorem for distributive justice
and we highlight several issues for future research.

Since the formulation of general equilibrium theory
\citep{arrow_existence_1954, debreu_theory_1959}, it is a canonical
proposition in economics that under certain assumptions decentralized
economic interaction---the most important instance of which are
markets---are efficient \citep{beckert_market_2002}. For instance, the
famed argument of \citet{Coase_ProblemSocialCost_1960} is that
``efficiency is inevitable in a world of complete information and
unrestricted contracting'' \citep[244]{Ray_CoalitionFormation_2014}.
However, the set of assumptions that ensure efficiency are either strict
or remarkably abstract, particularly with respect to the market
structure and the behavior of economic agents. Not surprisingly, a
substantial and sophisticated body of studies in economics examine
inefficiencies created by the deviations from the benchmark case of
equilibrium under rationality, complete information, complete contracts,
and perfect competition
\citep{pissarides_search_2001, stigler_economics_1961, akerlof_market_1970, grossman_impossibility_1980, simon_models_1982, rubinstein_modeling_1998, leibenstein_allocative_1966, williamson_economics_1981}.
In this paper, we begin our analysis from a concrete economic
interaction setting to consider the case where efficiency is a result of
institutions (i.e., rules of the game) or government intervention.

That is not an entirely strange thread of analysis. After all, it was
accepted wisdom in the early 20th century that decentralized economic
interaction, and thus markets, were inefficient \citep{gemici_neoclassical_2015}. The most influential proponents of markets, no less than \citet{schumpeter_theory_1983} for
instance, defended markets and capitalism with an appeal to innovation
and growth, not efficiency. Furthermore, influential streams of research
show that economic activities are embedded in and coordinated by
networks of interpersonal and organizational ties
\citep{granovetter_economic_1985, powell_neither_1990, gemici_karl_2008, krippner_elusive_2001}. Despite a recent
interest in economic networks \citep{jackson_social_2008}, the
efficiency of markets when their network structure is taken into account
is a theoretically and empirically understudied phenomenon. Precisely
because of this lacuna, the results in this paper have important
ramifications. Namely, if decentralized economic interaction depends on
intervention for efficiency and the distributive impacts of an
intervention such as tolls are not negligible, the {\em Iniquity
Theorem} raises questions related to distributive justice and social
contract. Furthermore, the issues of distribution and justice become
even more pressing, given that they result from purposeful intervention.
It is a common refrain and an empirically supported proposition that
capitalist markets are prone to the Matthew effect---``unto everyone
that hath shall be given, and he shall have abundance: but from him that
hath not shall be taken away even that which he hath''
\citep[3]{merton_matthew_1968}---where cumulative advantage and
preferential attachment result in the accumulation of resources and
positions by a relatively minor number of individuals and organizations
\citep{merton_matthew_1988, rigney_matthew_2010, wade_causes_2004}.
However, as the results in this paper indicate, that might not be the
only mechanism that exacerbates inequality. Intervention for the sake of
efficiency, which has been a guiding principle of economic policy in
many parts of the world since the 1970s
\citep{centeno_arc_2012, fourcade-gourinchas_rebirth_2002}, might also
play a crucial role in greater inequality.

We outline two issues pertaining to the relation of distribution and the
price of anarchy. The first, {\em the opportunity cost of iniquity},
addresses how iniquity affects the distribution of opportunities in a
particular market or in society. Here, the central question is the
impact of iniquity created by an intervention such as tolls on the
principle of equal opportunity. The second, {\em the cooperation effect
of iniquity}, examines whether the exacerbated inequalities ultimately
benefit the least advantaged in society and how that, in turn, affects
social cooperation. Here, apart from the normative questions, the issue
is whether iniquity makes social cooperation less likely.

Both of these issues emerge from Rawls's influential theory of justice,
which grounds justice on procedural and distributive fairness
\citep{RawlsJusticeFairnessRestatement2001, MandleRawlsTheoryJustice2009, FreemanJusticeSocialContract2009}.
This theory posits that, at the most fundamental level, policies and
institutions in society should guarantee basic liberties and
entitlements such as freedom of thought and the right to own property.
Once basic liberties and rights are secured, justice as fairness should
be based on (1) the fair equality of opportunity, and (2) the difference
principle. The fair equality of opportunity encapsulates the notion that
ascriptive factors such as socioeconomic background should not
constitute barriers to the attainment of positions in society. The
difference principle crystallizes the idea that social and economic
inequalities are justified to the extent that they are ``to the greatest
benefit of the least advantaged'' \citep[266]{RawlsTheoryJustice1999}.
It should be observed that even when it is unfeasible to achieve these
ideals, Rawls's two principles provide a benchmark to assess the extent
to which inequalities created by purposeful intervention are desirable.

% The Rawlsian idea of fair equality of opportunity is a procedural
% justice notion. It implies that not only positions should be open to all
% members of society, but also that the structure and rules of competition
% should be impartial \citep[73-76]{RawlsTheoryJustice1999}. The
% justification for the fair equality of opportunity comes from the fact
% that its absence amounts to the denial of the rights of autonomy and
% self-realization to some members of the society based on their
% background, which blatantly violates the notion of justice. However, one
% can also approach the problem from a different angle and examine the
% economic consequences of not having fair equality of opportunity.
The absence of equal opportunity means that some members of society will not
realize their full potential; at a minimum, it leads to a waste of human
resources and thus a sub-optimal use of one of the essential factors of
production. Hence, the absence of equal opportunity entails a cost,
which can be specified as the outcome, say in terms of innovation and
growth, that would materialize had the inequality-engendering
intervention not occurred
\citep{Banerjee_OccupationalChoiceProcess_1993}. We call this {\em the
opportunity cost of iniquity}. It should be noted that, although the
existing empirical evidence presents a complex picture, there is
sufficient support on the negative relationship between inequality and
growth, particularly when data from recent years are taken into account
\citep{Benhabib_TradeoffInequalityGrowth_2003, Cingano_TrendsIncomeInequality_2014}.
The relationship between innovation and inequality is murkier. Although
the existing empirical evidence lends support to the proposition that a
rapid period of innovation fosters inequality
\citep{Aghion_InnovationTopIncome_2015}, the opposite relationship is
difficult to decipher.

{\em The opportunity cost of iniquity} singles out the cost of iniquity
as compared to alternative trajectories that would occur under more
equitable distribution of opportunities and resources. The
{\em cooperation effect of iniquity} draws attention to the issue of
how greater inequality affects social cooperation. Rawls, following
Kantian contractual reasoning, suggests that fairness itself is a viable
and stable basis of social cooperation
\citep{ONeill_KantSocialContract_2012, Rawls_KantianConstructivismMoral_1980}.
The gist of his argument is that in a hypothetical world, a world where
individuals do not have certainty on their social positions and
distributional shares, an institutional rule such as ensuring that the
least advantaged receives the greatest benefit from a particular policy
is also the rational strategy for the members of society, simply because
it is what prudence requires under high uncertainty. To be sure, Rawls's
argument is an abstract thought experiment. However, it captures an
important issue. Namely, rules and policies that violate the principle
of fairness and entail greater inequality are also the rules and
policies that are likely to erode the bases of cooperation. We highlight
two analytically distinct aspects of this problem as potential avenues
for further inquiry:

\begin{itemize}
    \item {\em How does iniquity affect cooperation among members of society?} A voluminous and well-established literature examines the emergence and evolution of cooperation, often using a game-theoretic framework such as Iterated Prisoner's Dilemma \citep{Axelrod_FurtherEvolutionCooperation_1988, Axelrod_EvolutionCooperation_1981, Riolo_EvolutionCooperationReciprocity_2001, Bowles_CooperativeSpeciesHuman_2011}. Yet the effect of inequality on cooperation is a relatively novel area of research. The existing theoretical and empirical studies are in line with a Rawlsian emphasis on fairness as an important determinant of cooperation \citep{Fehr_ReciprocalFairnessCooperation_2002, Fehr_TheoryFairnessCompetition_1999, Colasante_ImpactInequalityCooperation_2014, Fehr_EconomicsFairnessReciprocity_2006}. However, given the recent theoretical results that show the complexity of cooperative behavior in a population when strategies and payoffs coevolve and when players have long memories \citep{Press_IteratedPrisonerDilemma_2012, Stewart_Collapsecooperationevolving_2014, Stewart_Smallgroupslong_2016}, further research is required to understand the evolution of cooperation under a trade-off between efficiency and equality.
    \item {\em How does iniquity affect the formation of groups and thus cooperation between different groups?} Apart from the general problem of cooperation, the trade-off between efficiency and equality also brings cooperation among different groups into focus. To use the language of cooperative game theory and political science, what is at stake is both coalition formation and stability of cooperation among coalitions once the trade-off between equality and efficiency are taken into account. It is true that the notion of fairness has been a fundamental element in cooperative game theory since the foundational works of the early 1950s \citep{Shapley_ValueNPersonGames_1953, Ray_CoalitionFormation_2014}. However, an explicit consideration of the trade-off between efficiency and fairness can offer fresh perspectives on the nature of cooperation among groups. 
\end{itemize}

\end{document}